\documentclass[11pt]{article}
\usepackage{graphicx}
\usepackage{fancyhdr}
\usepackage{fullpage}
\usepackage{booktabs} 
\usepackage{algorithm}
\usepackage[noend]{algpseudocode}
\usepackage[framemethod=tikz]{mdframed}
\usepackage[disable]{todonotes}
\usepackage{hyperref}
\usepackage{paralist}

\usepackage{amsmath,amssymb,amsfonts}
\usepackage{latexsym}
\usepackage{setspace}
\usepackage{xspace}
\usepackage{float}
\usepackage[caption = false]{subfig}

\newenvironment{proof}{\noindent{\bf Proof : \ }}{\hfill$\Box$\par\medskip}
\newtheorem{theorem}{Theorem}[section]

\newtheorem{lemma}[theorem]{Lemma}

\newtheorem{definition}[theorem]{Definition}

\newtheorem{problem}[theorem]{Problem}

\newenvironment{proofof}[1]{\begin{trivlist} \item {\bf Proof
#1:~~}}
  {\qed\end{trivlist}}
\renewenvironment{proofof}[1]{\par\medskip\noindent{\bf Proof of #1: \ }}{\hfill$\Box$\par\medskip}
\newcommand{\namedref}[2]{\hyperref[#2]{#1~\ref*{#2}}}
\newcommand{\thmlab}[1]{\label{thm:#1}}
\newcommand{\thmref}[1]{\namedref{Theorem}{thm:#1}}
\newcommand{\lemlab}[1]{\label{lem:#1}}
\newcommand{\lemref}[1]{\namedref{Lemma}{lem:#1}}

\newcommand{\applab}[1]{\label{app:#1}}
\newcommand{\appref}[1]{\namedref{Appendix}{app:#1}}

\newcommand{\figlab}[1]{\label{fig:#1}}
\newcommand{\figref}[1]{\namedref{Figure}{fig:#1}}
\newcommand{\alglab}[1]{\label{alg:#1}}
\renewcommand{\algref}[1]{\namedref{Algorithm}{alg:#1}}

\newcommand{\eqnlab}[1]{\label{eq:#1}}
\newcommand{\eqnref}[1]{\namedref{Equation}{eq:#1}}

\renewcommand{\O}[1]{\ensuremath{O\left(#1\right)}}

\newcommand{\margain}[2]{f\left( #1\, \middle| \, #2 \right)}
\def \lastbucket {\mdef{r}}
\def \greedy    {\mdef{\textsc{Greedy}}}
\def \offline    {\mdef{\textsc{Offline}}}
\def \algsize    {\mdef{\textsc{AlgSize}}}

\def \algnum    {\mdef{\textsc{AlgNum}}}
\def \algmult		{\mdef{\textsc{AlgMult}}}

\def \prune		{\mdef{\textsc{Prune}}}
\def \OPT    {\ensuremath{\mathcal {OPT}}\xspace}
\def \C    {\mdef{\textbf{C}}}
\def \x    {\mdef{\textbf{x}}}
\def \b    {\mdef{\textbf{b}}}
\def \numbuckets    {n}
\def \numpartitions    {\ell}
\def \capacity {t}
\def \extraspace {s}

\def \sieve		{\mdef{\textsc{Sieve}}}
\def \multidim		{\mdef{\textsc{Multidimensional}}}
\def \thresholding		{\mdef{\textsc{MarginalRatio}}}
\newcommand{\PPr}[1]{\ensuremath{\mathbf{Pr}\left[#1\right]}}

\newcommand{\eps}{\epsilon}

\newcommand{\mdef}[1]{{\ensuremath{#1}}\xspace}  % Math Def which can also be used in normal text.
\DeclareMathOperator*{\argmax}{argmax}

\def \etal{{\it et~al.}}

\newcommand{\ceil}[1]{\mdef{\left\lceil#1\right\rceil}}               % Absolute value
\newcommand{\set}[1]{\mdef{\left\{#1\right\}}}                        % Absolute value
\newcommand{\cost}{c}
\newcommand{\robust}{\mathcal S}

\newcommand{\mardens}[2]{\rho(#1|#2)}

\newcommand{\cA}{\mathcal{A}}
\newcommand{\cB}{\mathcal{B}}

\makeatletter
\renewcommand*{\@fnsymbol}[1]{\textcolor{blue}{\ensuremath{\ifcase#1\or *\or \dagger\or \ddagger\or
 \mathsection\or \triangledown\or \mathparagraph\or \|\or **\or \dagger\dagger
   \or \ddagger\ddagger \else\@ctrerr\fi}}}
\makeatother

\providecommand{\email}[1]{\href{mailto:#1}{\nolinkurl{#1}\xspace}}
\begin{document}

\title{Adversarially Robust Submodular Maximization under Knapsack Constraints}
\author{Dmitrii Avdiukhin\thanks{Indiana University. E-mail: \email{davdyukh@iu.edu}}
\and
Slobodan Mitrovi\'{c}\thanks{MIT. E-mail: \email{slobo@mit.edu}}
\and
Grigory Yaroslavtsev\thanks{Indiana University \& The Alan Turing Institute. E-mail: \email{gyarosla@iu.edu}}
\and
Samson Zhou\thanks{Indiana University. E-mail: \email{samsonzhou@gmail.com}}
}

\maketitle

\begin{abstract}
We propose the first adversarially robust algorithm for monotone submodular maximization under single and multiple knapsack constraints with scalable implementations in distributed and streaming settings. 
For a single knapsack constraint, our algorithm outputs a robust summary of almost optimal (up to polylogarithmic factors) size, from which a constant-factor approximation to the optimal solution can be constructed. 
For multiple knapsack constraints, our approximation is within a constant-factor of the best known non-robust solution. 

We evaluate the performance of our algorithms by comparison to natural robustifications of existing non-robust algorithms under two objectives: 1) dominating set for large social network graphs from Facebook and Twitter collected by the Stanford Network Analysis Project (SNAP), 2) movie recommendations on a dataset from MovieLens. 
Experimental results show that our algorithms give the best objective for a majority of the inputs and show strong performance even compared to offline algorithms that are given the set of removals in advance. 
\end{abstract}

\section{Introduction}
Submodular maximization has a wide range of applications in data science, machine learning and optimization, including data summarization, personalized recommendation, feature selection and clustering under various constraints, e.g. budget, diversity, fairness and privacy among others. 
\emph{Constrained} submodular optimization has been studied since the seminal work of~\cite{NemhauserWF78}. It has recently attracted a lot of interest in various large-scale computation settings, including distributed~\cite{BarbosaENW16,LiuV19}, streaming~\cite{BadanidiyuryMKK14,NorouziFardTMZMS18,AgarwalSS18}, and adaptive~\cite{GolovinK11,BalkanskiS18,BalkanskiRS19,FahrbachMZ19,EneN19, ChekuriQ19} due to its applications in recommendation systems \cite{LeskovecKGFVG07,ElAriniG11}, exemplar based clustering \cite{GomesK10}, and document summarization \cite{LinB11,WeiLKB13,SiposSSJ12}. 
For monotone functions, constrainted submodular optimization has been studied extensively under numerous constraints such as cardinality \cite{BadanidiyuryMKK14, BateniEM18}, knapsack \cite{HuangKY17}, matchings~\cite{ChakrabartiK14}, and matroids \cite{CalinescuCPV11}.

With the increase in volume of data, the task of designing low-memory streaming and low-communication distributed algorithms for monotone submodular maximization has received significant attention. 
A series of results~\cite{BadanidiyuryMKK14,NorouziFardTMZMS18,AgarwalSS18} culminated in a single pass algorithm over random-order stream that achieves close to $(1 - 1/e)$-approximation of monotone submodular maximization under cardinality constraint. 
This approximation almost matches the guarantee of a celebrated result~\cite{NemhauserWF78}. 
Also in the context of streaming, \cite{KumarMVV15,HuangKY17,YuXC18} studied monotone submodular maximization under $d$ knapsack constraints, resulting in a single pass algorithm that provides $(1 / (1 + 2d))$-approximation. Another line of work \cite{MirzasoleimanKSK13,KumarMVV15,BarbosaENW16,LiuV19} focused on submodular maximization in distributed setting. 
In particular, \cite{LiuV19} developed $2$-round and $(1/\eps)$-round MapReduce algorithms that provide $1/2$ and $1 - 1/e - \eps$ approximation, respectively, for monotone submodular maximization under cardinality constraint.

In this paper, we focus on the robust version of this classic problem~\cite{KrauseMGG08,OrlinSU18,MitrovicBNTC17,KazemiZK18}. 
Consider a situation where a set of recommendations (or advertisements) is constructed for a new user. 
It is standard to model this as a monotone submodular maximization problem under knapsack constraints, which allow incorporation of various restrictions on available budget, screen space, user preferences, privacy and fairness, etc. 
However, new users are likely to find some of the recommended items familiar, annoying or otherwise undesirable. 
Hence, it is advisable to build recommendations in such a way that even if the user later decides to dismiss some of the recommended items, one can quickly compute a new high-quality set of recommended items without solving the entire problem from scratch.
We refer to this property as ``adversarial robustness'' since the removals are allowed to be completely arbitrary (e.g. might depend on the algorithm's suggestions).

\subsection{Adversarially Robust Monotone Submodular Maximization}
\label{sec:rmsm}
Let $V$ be a finite domain consisting of elements $e_1, \dots, e_{|V|}$. 
For a set function $f \colon 2^{V} \to \mathbb R^{\ge 0}$, we use $\margain{e}{S}$ to denote the \textit{marginal gain} of an element $e$ given a set $S \subseteq V$, i.e., $\margain{e}{S} = f(S \cup \{e\}) - f(S)$. 
A set function $f$ is \textit{submodular} if for every $S \subseteq T \subseteq V$ and every $e \in V$ it holds that $\margain{e}{T} \le \margain{e}{S}$. 
A set function $f$ is \textit{monotone} if for every $S \subseteq T \subseteq V$ it holds that $f(T) \ge f(S)$.
Intuitively, elements in the universe contribute non-negative utility, but have diminishing gains as the cost of the set increases. 

For a set $S \subseteq V$ we use notation $\x_S$ to denote the 0-1 indicator vector of $S$. 
We use $\C \in \mathbb R^{d \times |V|}$ to denote a matrix with positive entries and $\b \in \mathbb R^{d}$ to denote a vector with positive entries. 
Here, $\C$ and $b$ should be interpreted as knapsack constraints, where set $S$ satisfies these constraints if and only if $\C \x_S \le b$.

\begin{problem}[MSM under knapsack constraints]\label{def:msm}
In the \textit{monotone submodular maximization (MSM) problem} subject to $d$ knapsack constraints, we are given a monotone submodular set function $f \colon 2^V \to \mathbb R^{\ge 0}$ and are required to output: $$\OPT(V) = \argmax_{S \subseteq V \colon \C \x_S \le \b} f(S).$$ 
\end{problem}

Since the constraints are scaling-invariant, one can rescale each row $\C_i$  by multiplying it (and the corresponding entry in $\b$) by $b_1 / b_i$ so that all entries in $\b$ are the same and equal to $b_1$. 
One can further rescale $\C$ and $\b$ by the smallest entry in $\C$ (or some lower bound on it), so that $\min_{i,j} C_{i,j} \ge 1$. 
We assume such rescaling below and let $K = b_i$ for all $i$.
In the case of one constraint ($d = 1$), we further simplify the notation and set $\cost(e_i) = C_{1,i}$ and $K = b_1$ and refer to $\cost(e_i)$ simply as the cost of the $i$-th item.

An important role in our algorithms is played by the \textit{marginal density} of an item.
Formally, for a set $S \subseteq V$, an element $e$ and a cost function $c \colon V \to \mathbb R^{\ge 0}$ we define the marginal density of $e$ with respect to $S$ under the cost function $c$ as:
$\mardens{e}{S} = \frac{\margain{e}{S}}{c(e)}$. 
For multiple dimensions, we will specifically define the cost function $c(\cdot)$.

Motivated by applications to personalized recommendation systems, we consider the \textit{adversarially robust} version of the above problem. 
In the \textit{adversarially robust monotone submodular maximization (ARMSM) problem} the goal is to produce a small ``adversarially robust'' summary $\robust \subseteq V$. 
Here ``adversarial robustness'' means that for any set $E$ of cardinality at most $m$, which might be later removed, one should be able to compute a good approximation for the residual monotone submodular maximization problem over $V \setminus E$ based only on $\robust$. 
In this paper, we propose a study of ARMSM under knapsack constraints:

\begin{problem}[ARMSM under knapsack constraints]\label{def:rmsm} 
An algorithm $\mathcal A$ solves the adversarially robust monotone submodular maximization problem ARMSM$(m,K)$ subject to $d$ knapsack constraints if it produces a summary $\robust \subseteq V$ such that:
$$\OPT(V \setminus E) = \argmax_{S \subseteq \robust \setminus E \colon \C \x_S \le \b} f(S)$$
for any set of removals $E$ of cardinality at most $m$.
$\mathcal A$ gives an $\alpha$-approximation if there exists a set $Z\subseteq\robust$ with $\C\x_Z\le\b$ such that $f(Z) \ge \alpha f(\OPT (V \setminus E))$.
\end{problem}
The main goal of an adversarially robust algorithm is to minimize the size of the resulting summary. 
We remark that the above robustness model is very strong. 
In particular, the set of removals $E$ does not have to be fixed in advance and might depend on the summary $\robust$ produced by the algorithm. 
Hence, we choose to refer to it as \textit{adversarial robustness} in order to avoid confusion with other notions of robustness known in the literature \cite{AnariHNPST17, StaibWJ18}.

\subsection{Our Theoretical Results}
\paragraph{Streaming algorithms.} 
We first consider the ARMSM problem in the streaming setting. A streaming algorithm is given the vector $\b$ of knapsack budget bounds upfront. 
Then, the elements of the ground set $e_1, \dots, e_{|V|}$ arrive in an arbitrary order. 
When an element $e_i$ arrives, the algorithm sees the corresponding column $\C_{*,i}$, which lists the $d$ costs associated with this item. 
The algorithm only sees each element once and is required to use only a small amount of space throughout the stream. 
In the end of the stream, an adversarially chosen set of removals $E$ is revealed and the goal is to solve ARMSM over $V \setminus E$. 
The key objective of the streaming algorithm is to minimize the amount of space used while providing a good approximation for ARMSM for any $E$.

Our first set of results gives adversarially robust algorithms for the ARMSM problem under one knapsack constraint:

\begin{theorem}[ARMSM under one knapsack constraint]\thmlab{thm:one-knapsack}
For the ARMSM$(m,K)$ problem under one knapsack constraint, there exists an algorithm that gives a constant-factor approximation with a summary consisting of $\tilde O(K + m)$ elements of the ground set (\thmref{thm:items}). 
\end{theorem}
We also show that if the total cost of removed items is at most $M$ then there is an algorithm with summary size $\tilde O(K + M)$ and improved approximation guarantee. 
For ARMSM under a single knapsack constraint, our bounds are tight up to polylogarithmic factors, since an optimal solution may contain $K$ items of unit cost, and an adversary can remove up to $m$ items of any set. 
Hence, storing $\Omega(K + m)$ elements is necessary to obtain a constant factor approximation. 

For the ARMSM problem under $d$ knapsack constraints, we give an algorithm with the following guarantee:
\begin{theorem}[ARMSM under $d$ knapsack constraints]\thmlab{thm:d-knapsacks}
For the ARMSM$(m,K)$ problem under $d$ knapsack constraints, there exists an algorithm that gives an $\Omega(\frac1d)$-approximation with a summary of size $\tilde O(K + m)$ (\thmref{thm:mult}). 
\end{theorem}

\paragraph{Distributed algorithms.} 
We also consider the ARMSM problem in the distributed setting. 
Here, our aim is to collect a robust set $\robust$ of elements while distributing the work to a number of machines, minimizing the memory requirement per machine and the number of rounds in which the machines need to communicate with each other. 
As in the case of streaming setting, a set of removals $E$ is revealed only after $\robust$ is constructed. 
We obtain a $2$-round algorithm that matches our result for streaming, in terms of approximation guarantees.
\begin{theorem}[Distributed ARMSM]
For the ARMSM$(m,K)$ problem on a dataset of size $n$ under $d$ knapsack constraints, there exists an algorithm that gives an $\Omega(\frac1d)$-approximation with a summary of size $\tilde O(K + m)$. 
If oracle access to $f$ is given, this algorithm can be implemented in two distributed rounds using $\tilde O((m+K)\sqrt{n})$ words of space per machine (\thmref{thm:distributed}).
\end{theorem}

\subsection{Empirical Evaluations}
We evaluate the performance of our algorithms on both single knapsack and multiple knapsack constraints by comparison to natural generalizations of existing algorithms. 
We implement the algorithms for the objective of dominating set for large social network graphs from Facebook and Twitter collected by the Stanford Network Analysis Project (SNAP), and for the objective of coverage on a large dataset from MovieLens. 
We compare the objectives on the sets output as well as the total number of elements collected by each algorithm. 

Our results show that our algorithms provide the best objective for a majority of the inputs. 
In fact, our streaming algorithms perform just as well as the standard \emph{offline} algorithms, even when the offline algorithms \emph{know in advance} which elements will be removed. 
Our results also indicate that the number of elements collected by our algorithms does not appear to correlate with the total number of elements, which is an attractive property for streaming algorithms. 
In fact, most of the baseline algorithms collect relatively the same number of elements for the robust summary, ensuring fair comparison. 
For more details, see Section \ref{sec:experiments}.

\subsection{Previous Work}
The special case of ARMSM$(m,K)$ with one constraint and equal costs for all elements is referred to as robust submodular maximization under the cardinality constraint. 
If at most $k$ elements can be selected, we refer to this problem as ARMSM$(m,k)$. 
The study of this problem was initiated by Krause\,\etal~\cite{KrauseMGG08}.
The first (non-streaming) constant-factor approximation for this problem was given by Orlin\,\etal~\cite{OrlinSU18} for $m = o(\sqrt{k})$.
This was further extended by~\cite{BogunovicMSC17} who give algorithms for $m = o(k)$. 
In these works, the size of the summary is restricted to contain at most $k$ elements and hence by design only $m < k$ removals can be handled.

Recently the focus has shifted to handling larger numbers of removals and so there has been increased interest in studying ARMSM$(m,k)$ with summary of sizes \emph{greater} than $k$. 
\cite{MirzasoleimanKK17} solve this problem with summary size $O(k \cdot m)$, which was improved by \cite{MitrovicBNTC17} to $\tilde{O}(m + k)$. 
Moreover, their algorithms are applicable to arbitrary ordered streams. 
A different setup was considered by \cite{KazemiZK18}, who assume that $E$ is chosen \emph{independently} of the choice of a robust summary and give algorithms with summary size $\tilde{O}(m + k)$, but obtain better approximation guarantees than \cite{MitrovicBNTC17}.

To the best of our knowledge, there is little known about the general ARMSM$(m,K)$ problem considered here which asks for robustness under single or multiple knapsack constraints.

\section{Techniques}
Our general approach is to find a set $\robust$ at the end of the stream, so that when a set $E$ of items is removed, we show that running an offline algorithm, \offline, on the set $Z:=\robust\setminus E$ produces a good approximation to the value of the optimal solution of the entire stream. 
Since \offline{} on input $Z$ is known to produce a good approximation to the optimal solution of constrained submodular maximization on input $Z$ (see \thmref{thm:offline}), then it suffices to show that $f(\OPT(Z))$ is a good approximation to $f(\OPT)$, where we use $\OPT$ to denote $\OPT(V\setminus E)$.

\begin{theorem}
\thmlab{thm:offline}
\cite{Sviridenko04,KulikST13}
There exists an algorithm \offline{} that gives a $(1 - 1/e)$-approximation for the monotone submodular maximization problem subject to $d$ knapsack constraints in polynomial time. 
\end{theorem}

We assume that we have a good guess for $f(\OPT)$ by making a number of exponentially increasing guesses $\tau$. 
Our algorithms start with the partitions-and-buckets approach from \cite{BogunovicMSC17, MitrovicBNTC17} for robust submodular maximization under cardinality constraints. 
Specifically, our algorithms create a number of partitions and also create a number of buckets for each partition, where the number of buckets is chosen to be ``robust'' to the removal of items at the end of the stream. 
An element in the stream is added to the first possible bucket in which its marginal density exceeds a certain threshold, which is governed by the partition. 
The thresholds are exponentially decreasing across the partitions, so that the number of partitions is logarithmic in $K$. 

At a high level, our algorithms overcome several potential pitfalls. 
The first challenge we face is the issue of buckets being populated with items of small cost whose marginal density surpasses the threshold. 
These small items prevent large items (such as cost $K$) whose marginal density also surpasses the threshold from being added to any bucket. 
If the optimal solution consists of a single large item, then the approximation guarantee could potentially be as bad as $\frac{1}{K}$. 
Thus, we allow each bucket \emph{double} the capacity and create an additional partition level with a smaller threshold to compensate. 

The second challenge we face is relating the items in various partitions. 
Although we would like to argue that an item $e$ in a bucket in a certain partition $i$ does not have overwhelmingly large marginal gain, the most natural way to prove this would be to claim that $e$ would have been placed in a previous partition less than $i$ because the ratio is overwhelmingly large. 
However, this is no longer true because items in partition $i$ can have up to cost $2^i$ and any non-empty bucket in previous partitions does not have enough capacity. 
Surprisingly, for the purposes of analysis, it suffices to prohibit any item in a bucket from using more than a certain fraction of the capacity. 
That is, any item added to a bucket $B_{i,j}$, which has capacity $2^{i+1}$, must have cost at most $2^{i-1}$. 

\subsection{Robustness to the Removal of $m$ Items}
We now describe \algnum, which outputs a solution of cost $K$ on a single knapsack constraint and is robust against the removal of up to $m$ items. 
We would like to use an averaging arguments to show that some ``saturated'' bucket $B_{i^*,j}$ in a partition cannot have too much intersection with the elements $E$ that are removed at the end of the stream. 
However, the removal of up to $m$ items at the end of the stream may cause the removal of cost up to $mK$. 
But then the averaging argument fails unless the number of buckets in each partition also increases by a factor of $K$, which unfortunately gives an additional multiple of $K$ in the space of the algorithm. 

Instead, the key idea is to \emph{dynamically} allocate a number of new buckets, depending on the total cost of the current items in the buckets of a partition. 
The goal is to maintain enough buckets to guarantee that a certain number of elements can be added to a partition, regardless of their cost. 
Therefore, the number of total buckets is not large unless the stored items have large cost, in which case the number of items is relatively low anyway. 
To do this, we maintain counters $\extraspace_i$ that allocate a new bucket to partition $i$ each time they exceed $\min\{2^i,K\}$. 
Each time an item $e$ is added to partition $i$, the counter $\extraspace_i$ is increased proportional to the cost of the item, $c(e)$. 
The creation of new buckets is allowed until a certain number of items have been collected by the partition. 
Intuitively, algorithms robust to the removal of $m$ items, such as \algnum, should strive to output at the end of the stream a set with a certain number of items, whereas algorithms robust to the removal of items with a certain cost $M$ should strive to output a set with a certain cost. 

At the end, we run a procedure \prune{} to further bound the number of elements output by the algorithm. 
\prune{} simply reorders the elements stored by \algnum{} by cost of the elements, and again runs \algnum{} on the sorted set of elements as an input stream. 
Since the items with smaller cost arrive first, this ensures that we cannot have too many items of large cost. 
\section{Streaming Algorithms}
We now warm-up by providing the first streaming algorithm for the ARMSM$(m,K)$ problem under a single knapsack constraint. 
We later show how to build on these ideas to obtain robustness subject to multiple knapsack constraints.
\subsection{Single Knapsack Constraint}
\label{sec:single}
We describe our algorithm \algnum, which is used to produce a summary consisting of $\tilde O(K + m)$ items. 
Recall that we use $\OPT$ to denote $\OPT(V\setminus E)$.
In order to simplify presentation, we assume\footnote{This assumption can be removed using standard techniques (see e.g. Appendix E of \cite{MitrovicBNTC17}) by maintaining $\frac{1}{\eps}\log K$ guesses to find such a $\tau^*$. } that we have a good estimate $\tau^*$ for $f(\OPT)$, such that $\tau^*\le f(\OPT)\le (1+\eps)\tau^*$. 
To simplify presentation, we further assume that $K$ is a power of two and hence let $K = 2^{\ell}$ (see \algref{alg:algnum} for how rounding is handled).

\algnum{} creates $\numpartitions$ partitions $B_1, \dots, B_\numpartitions$ where the $i$-th partition initially consists of $n_i = O(\ell (\frac{m}{2^i} + 1))$ buckets of capacity $2^{i + 1}$ each. 
We refer to the $j$-th bucket in the $i$-th partition as $B_{i,j}$. 
When processing the stream, each element $e$ is added to the first possible bucket $B_{i,j}$ in the first possible partition $i$ such that the bucket has enough capacity remaining and the marginal density $\mardens{e}{B_{i,j}}$ exceeds a certain threshold $\tau/2^i$ for this partition. 
%For the sake of analysis, when considering the elements we also skip partition levels with buckets of capacity less than half of the item's cost. 
Note that the thresholds exponentially decrease across the partitions while capacities of the buckets exponentially increase.

Our goal is to maintain enough buckets to guarantee that a certain number of elements can be added to a partition, regardless of their cost. 
To dynamically allocate a number of new buckets, \algnum{} keeps counters $\extraspace_i$ that create a new bucket while they exceed $2^i$, after which the value of the counter is lowered. 
The counter $\extraspace_i$ is increased proportional to the cost of an item $c(e)$ each time an item $e$ is added to partition $i$. 
This process continues until a certain number of items are in the partition. 
Finally, we run the procedure \prune{} to further bound the number of elements output by the algorithm. 
See \figref{fig:buckets} for an illustration of the data structure.

\begin{algorithm}[!htb]
\caption{\algnum: Picking elements with large marginal density.}\alglab{alg:algnum}
\begin{algorithmic}[1]
\Require{Parameters $m, K$, estimate $\tau^*$ of $f(\OPT)$.}
\State{$\numpartitions \leftarrow \ceil{\log K}, w \leftarrow \ceil{\frac{4\numpartitions m}{K}}$, $\tau\leftarrow\frac{2\tau^*}{32\left(1-\frac{1}{2\numpartitions}\right)+3}$}
\For{$i \leftarrow 0$ to $\numpartitions$}
\Comment{Initialize parameters}
	\State{$\numbuckets_i\leftarrow w\ceil{K/2^i}+8\numpartitions$}
	\State{$\extraspace_i\leftarrow 0$}
	\For{$j \leftarrow 1$ to $\numbuckets_i$}
	\State{$B_{i,j}\leftarrow\emptyset$} 
\EndFor
\EndFor	
\For{each element $e$ in the stream}
\For{$i\leftarrow 0$ to $\numpartitions$}
\If {$c(e) > 2^{i-1}$} \textbf{continue}
\EndIf
\For{$j\leftarrow 1$ to $\numbuckets_i$}
\If{$\mardens{e}{B_{i,j}} < \frac{\tau}{2^i}$} \textbf{continue} 
\EndIf
\If{$c(B_{i,j} \cup e)\le 2^{i+1}$ }
\State{$B_{i,j}\leftarrow B_{i,j}\cup\{e\}$}
\State{$\extraspace_i\leftarrow \extraspace_i+8\numpartitions c(e)$}
\If {$\sum_{j=1}^{\numbuckets_i}|B_{i,j}|<10w\cdot 2^i$}
\While{$\extraspace_i\ge 2^i$} 
\State{$B_{i,\numbuckets_i+1}\leftarrow\emptyset$}
\State{$\numbuckets_i\leftarrow\numbuckets_i+1$}
\State{$\extraspace_i\leftarrow \extraspace_i-2^i$}
\EndWhile
\EndIf
\State{\textbf{break:} process next element $e$}
\EndIf
\EndFor
\EndFor
\EndFor\\
\Return $S_\tau = \{B_{i,j}\}_{i,j}$
%, Z_\tau = \offline(\bigcup_{i,j} B_{i,j}\setminus E)$
\end{algorithmic}
\end{algorithm}

\begin{algorithm}[!htb]
\caption{\prune: Decreasing the size of the output set.}\alglab{alg:prune}
\begin{algorithmic}[1]
\Require Output set $S$ from \algnum.
\State{Sort $S$ by size so that $c(s_1)\le c(s_2)\le\ldots$.}
\State{$T\leftarrow\algnum$ with input stream $s_1,s_2,\ldots$.}\\
\Return $T$
\Comment{\algnum{} can be replaced with \algmult.}
\end{algorithmic}
\end{algorithm}

\begin{algorithm}[!htb]
\caption{Robust maximum submodular with knapsack constraint.}
\begin{algorithmic}[1]
\Require Output set $T$ from \prune, set $E$ removed by adversary.
\\\Return $\offline(T\setminus E)$
\end{algorithmic}
\end{algorithm}

\begin{figure*}[htb]
\centering
\begin{tikzpicture}[scale=1]
\node at (-2,0.5){Partition 0:};
\node at (-0.5,0.5){$\tau$};
\filldraw[shading=radial,inner color=white, outer color=red!50!, opacity=1] (0,0) rectangle +(0.5,1);
\filldraw[shading=radial,inner color=white, outer color=green!50!, opacity=1] (0.5,0) rectangle +(0.5,1);
\draw (0,0) rectangle +(1,1);
\draw[dashed] (0.5,0) -- (0.5,1);
\filldraw[shading=radial,inner color=white, outer color=yellow!50!, opacity=1] (1.1,0) rectangle +(0.5,1);
\filldraw[shading=radial,inner color=white, outer color=orange!50!, opacity=1] (1.6,0) rectangle +(0.5,1);
\draw (1.1,0) rectangle +(1,1);
\filldraw[shading=radial,inner color=white, outer color=purple!50!, opacity=1] (2.2,0) rectangle +(0.5,1);
\draw (2.2,0) rectangle +(1,1);
\filldraw[shading=radial,inner color=white, outer color=blue!50!, opacity=1] (3.3,0) rectangle +(0.5,1);
\filldraw[shading=radial,inner color=white, outer color=cyan!50!, opacity=1] (3.8,0) rectangle +(0.5,1);
\draw (3.3,0) rectangle +(1,1);
\filldraw[shading=radial,inner color=white, outer color=pink!50!, opacity=1] (4.4,0) rectangle +(0.5,1);
\draw (4.4,0) rectangle +(1,1);
\draw (5.5,0) rectangle +(1,1);
\draw (6.6,0) rectangle +(1,1);
\draw (7.7,0) rectangle +(1,1);
\node at (0.5,0.5){$2$};
\node at (1.1+0.5,0.5){$2$};
\node at (2.2+0.5,0.5){$2$};
\node at (3.3+0.5,0.5){$2$};
\node at (4.4+0.5,0.5){$2$};
\node at (5.5+0.5,0.5){$2$};
\node at (6.6+0.5,0.5){$2$};
\node at (7.7+0.5,0.5){$2$};

\node at (-2,1.6){Partition 1:};
\node at (-0.5,1.6){$\tau/2$};
\filldraw[shading=radial,inner color=white, outer color=blue!50!, opacity=1] (0,1.1) rectangle +(0.5,1);
\filldraw[shading=radial,inner color=white, outer color=red!50!, opacity=1] (0.5,1.1) rectangle +(1,1);
\draw (0,1.1) rectangle +(2,1);
\filldraw[shading=radial,inner color=white, outer color=cyan!50!, opacity=1] (2.2,1.1) rectangle +(0.5,1);
\filldraw[shading=radial,inner color=white, outer color=lime!50!, opacity=1] (2.7,1.1) rectangle +(0.5,1);
\draw (2.2,1.1) rectangle +(2,1);
\draw (4.4,1.1) rectangle +(2,1);
\draw (6.6,1.1) rectangle +(2,1);
\node at (1,1.1+0.5){$4$};
\node at (2.2+1,1.1+0.5){$4$};
\node at (4.4+1,1.1+0.5){$4$};
\node at (6.6+1,1.1+0.5){$4$};

\node at (-2,2.7){Partition 2:};
\node at (-0.5,2.7){$\tau/4$};
\filldraw[shading=radial,inner color=white, outer color=green!50!, opacity=1] (0,2.2) rectangle +(2,1);
\filldraw[shading=radial,inner color=white, outer color=purple!50!, opacity=1] (2,2.2) rectangle +(2,1);
\draw (0,2.2) rectangle +(4,1);
\filldraw[shading=radial,inner color=white, outer color=violet!50!, opacity=1] (4.4,2.2) rectangle +(2,1);
\draw (4.4,2.2) rectangle +(4,1);
\draw [thick,red] (8.8,2.2) rectangle +(4,1);
\node at (2,2.2+0.5){$8$};
\node at (4.4+2,2.2+0.5){$8$};
\node [red] at (8.8+2,2.2+0.5){$8$};

\node at (2.2,3.9){$\vdots$};
\node at (6.6,3.9){$\vdots$};

\node at (0.5,-0.5){$B_{0,1}$};
\node at (1.1+0.5,-0.5){$B_{0,2}$};
\node at (2.2+0.5,-0.5){$B_{0,3}$};
\node at (3.3+0.5,-0.5){$B_{0,4}$};
\draw[->, black] (0.3,-0.8) -- (8.5,-0.8);
\draw[->, black] (-3.2,0) -- (-3.2,3.8);
\node at (-0.5,3.8){Threshold};
\end{tikzpicture}
\caption{Partitions and buckets: $B_{1,1}$ and $B_{2,2}$ are partially occupied. $B_{2,3}$ (in red) has been dynamically created, since the elements of $B_{2,1}$ and $B_{2,2}$ are large.}\figlab{fig:buckets}
\end{figure*}
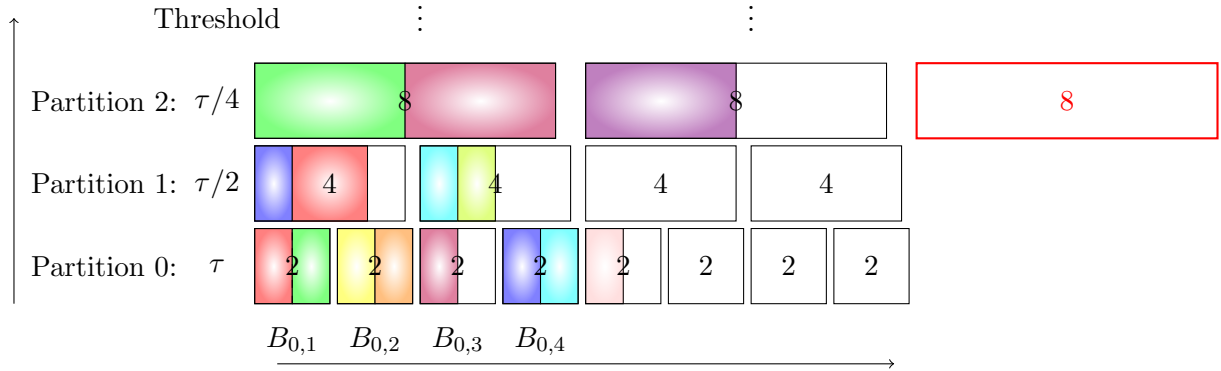

By using \prune{} on the output of \algnum, we have the following result, whose proof formally appears in Appendix~\ref{app:missing:single}.
\begin{theorem}
\thmlab{thm:items}
%Let $\numpartitions=\ceil{\log K}$ and $\zeta=1-\frac{1}{2\numpartitions}$. 
There exists an algorithm that outputs a set $\robust$ with $\tilde{O}(K+m)$ elements such that, for any set $E$ of at most $m$ removed items, one can compute from $\robust$ a set $Z \subseteq V \setminus E$ with cost at most $K$ and $f(Z)$ is a constant factor approximation to $f(\OPT)$.
%, where $r=\frac{2 (1 - 1/e)\zeta}{32\zeta+3}-\eps$.  
%As $K\rightarrow\infty$, the approximation ratio is at least $0.036-\eps$. 
\end{theorem}
In fact, if the $m$ items that are removed has total cost at most $M$, we can provide a better guarantee in terms of both approximation and number of elements stored (see~\appref{app:size}). 

\subsection{$d$ Knapsack Constraints}
\label{sec:d-knapsacks}
We now consider the ARMSM$(m,K)$ problem under $d$ knapsack constraints. 
Recall that \algnum{} relies on guessing the correct threshold and then using a streaming framework that adds elements whose marginal gain surpasses the threshold. 
In the case where there are $d$ knapsack constraints, a natural approach would be to have parallel instances that guess thresholds for each constraint, and then pick the instance with the best set. 
This would certainly work, but since there would be $\O{\log K}$ guesses for each constraint, the total number of parallel instances would be $\O{\log^d K}$, which is unacceptable for large values of $d$ and $K$. 
On the other hand, it seems reasonable to believe that the space usage can be improved, at the expense of the approximation guarantee, by maintaining a smaller number of parallel instances. 
In that case, marginal gain to cost ratio is not well-defined, since there is a separate cost for each knapsack, so what would be the right quantity to consider?

Recall the standard normalization for multiple knapsack constraints discussed in Section~\ref{sec:rmsm}. 
We define the largest cost of an item to be the maximum cost of the item across all knapsacks, after the normalization. 
It has been previously shown that the correct quantity to consider for the streaming model is the marginal gain of an item divided by its largest cost \cite{YuXC18}. 
Namely, if the ratio of the marginal gain to the largest cost of an item exceeds the corresponding threshold, and the item fits into a bucket without violating any of the knapsack constraints, then we choose to add the item to the first such bucket. 
Since the threshold now compares the marginal gain to the largest cost, a natural question to be asked is what quantity should be used for the dynamic allocation of the buckets. 
Recall that the previous goal of \algnum{} was to maintain a specific number of items, so that it would be robust against the removal of $m$ items.
Thus, we would like to allocate a new bucket for a partition whenever the capacity of the bucket with respect to some knapsack becomes saturated. 
Hence, \algmult{} maintains a series of counters $\space_{i,a}$ for partition $i$ and knapsack $a$, where $1\le a\le d$. 
Whenever one of these counters exceeds $K$, we create a new bucket entirely in partition $i$, and lower $\space_{i,a}$ accordingly.

\begin{algorithm}[!htb]
\caption{\algmult: Picking elements with large marginal gain to cost ratio.}
\begin{algorithmic}[1]
\Require{Parameters $d, m, K$, estimate $\tau^*$ of $f(\OPT)$.}
\State{$\numpartitions \leftarrow \ceil{\log K}, w \leftarrow \ceil{\frac{4\numpartitions m}{K}}$, $\tau\leftarrow\frac{\tau^*}{4}$}
\For{$i \leftarrow 0$ to $\numpartitions$}
\Comment{Initialize parameters}
	\State{$\numbuckets_i\leftarrow w\ceil{K/2^i}+8\numpartitions$}
	\State{$\extraspace_{i,a}\leftarrow 0$}
	%\State{$\capacity_i = \min\{2^i, K\}$}
	\For{$j \leftarrow 1$ to $\numbuckets_i$ \label{line:algmult-reset-B}}
	\State{$B_{i,j}\leftarrow\emptyset$} 
\EndFor
\EndFor	
\For{each element $e$ in the stream}
\State{$c(e)\leftarrow\underset{1\le a\le d}{\max}c_a(e)$}
\For{$i\leftarrow 0$ to $\numpartitions$}
\If {$c(e) > 2^{i-1}$} \textbf{continue}
\EndIf
\For{$j\leftarrow 1$ to $\numbuckets_i$}
\If{$\rho(e | B_{i,j}) < \frac{\tau}{2^i(1+2d)}$} \textbf{continue} 
\EndIf
\If{$c_a(B_{i,j} \cup e)< 2^{i+1}$ for all $1\le a\le d$}
\State{$B_{i,j}\leftarrow B_{i,j}\cup\{e\}$}
\State{$\extraspace_{i,a}\leftarrow \extraspace_{i,a}+8\numpartitions c_a(e)$ for all $1\le a\le d$}
\While{$\extraspace_{i,a}\ge 2^i$ for some $1\le a\le d$ and $\sum_{j=1}^{\numbuckets_i}|B_{i,j}|<10w\cdot 2^i$} 
\State{$B_{i,\numbuckets_i+1}\leftarrow\emptyset$}
\State{$\numbuckets_i\leftarrow\numbuckets_i+1$}
\State{$\extraspace_{i,a}\leftarrow\max\{0,\extraspace_{i,a}-2^i\}$ for all $1\le a\le d$}
\EndWhile
\State{\textbf{break:} process next element $e$}
\EndIf
\EndFor
\EndFor
\EndFor\\
\Return $S_\tau = \{B_{i,j}\}_{i,j}$
%, Z_\tau = \offline(\bigcup_{i,j} B_{i,j}\setminus E)$
\end{algorithmic}
\end{algorithm}

By using \prune{} on the output of \algnum, we have the following result, whose proof formally appears in Appendix~\ref{app:missing:mult}.
%Let $Z = \offline\left(\bigcup_{B_{i,j} \in S} B_{i,j}\setminus E\right)$, where $E$ is the set of elements that are removed at the end of the stream. 
As in Section~\ref{sec:single}, we do not attempt to optimize parameters here, but observe that the number of elements stored is independent of $d$. 
\begin{theorem}
\thmlab{thm:mult}
For the ARMSM$(m,K)$ problem under $d$ knapsack constraints, there exists an algorithm that outputs a set $\robust$ of size $\tilde{O}(K+m)$, from which one can compute a set $Z\subseteq V\setminus E$ with cost at most $K$ and $f(Z)$ is a $\Omega\left(\frac{1}{d}\right)$-approximation to $f(\OPT)$.
%Let $\gamma=1-\frac{1}{e}$. 
%For the ARMSM$(m,K)$ problem under $d$ knapsack constraints, there exists an algorithm that outputs a set $S$ of size $\O{\frac{1}{\eps}(K\log^3+m\log^4 K)}$, from which one can compute a set $Z\subseteq V\setminus E$ with cost at most $K$ and $f(Z)$ is an $r^2$-approximation to $f(\OPT)$, where $r=\left(\frac{\gamma}{32d}\left(1-\frac{1}{4(2d+1)}\right)-\eps\right)$. 
\end{theorem}
\section{Distributed Algorithm}
\label{sec:distributed}
\newcommand{\memory}{L}
\newcommand{\machines}{T}
In this section, we give a distributed algorithm for the ARMSM$(m,K)$ problem under $d$ knapsack constraints (see Definition~\ref{def:rmsm}). 
We use a variant of the MapReduce model of \cite{KarloffSV10}, in which we consider an input set $V$ of size $n=|V|$ that is distributed across $\tilde{O}((K+m)\sqrt{n})$ machines. 
For some parameters $m$ and $K$ that are known across all machines, we permit each machine to have $\tilde{O}((K+m)\sqrt{n})$ memory. 
The machines communicate to each other in a number of synchronous rounds to perform computation. 
In each round, each machine receives some input of size $\tilde{O}((K+m)\sqrt{n})$, on which the machine performs some local computation. 
The machine then communicates some output to other machines at the start of the next round. 
We require that the total input and output message size is $\tilde{O}((K+m)\sqrt{n})$ per machine. 
We assume that each machine has access to an oracle that computes $f$.
Then our main result in the distributed model is the following.
\begin{theorem}
\thmlab{thm:distributed}
For the ARMSM$(m,K)$ problem under $d$ knapsack constraints, there exists a two-round distributed algorithm that outputs a set $\robust$, from which one can compute a set $Z\subseteq V\setminus E$ with cost at most $K$ and $f(Z)$ is a $\Omega\left(\frac{1}{d}\right)$-factor approximation to $f(\OPT)$. 
Moreover, each machine uses space $\tilde{O}\left(\sqrt{n \cdot \frac{1}{\eps}(K^2+mK)}\right)$.
\end{theorem}

The analysis of our distributed algorithm is based on the analysis for our streaming algorithms, along with a recent work by \cite{LiuV19}. 
We generalize their result to obtain a distributed algorithm that constructs a robust summary equivalent to that constructed by $\algmult$.

In our algorithm and proofs, we use $\memory$ to denote an upper bound on the number of elements collected by \algmult. 
Let $\cB$ be the data structure of sets $B_{i, j}$ maintained by \algmult. 
We use $\algmult_{\cB, W}$ to refer to the invocation of \algmult with the following changes:
\begin{itemize}
\item The buckets $B_{i, j}$ are initialized by $\cB$ and the loop on line~\ref{line:algmult-reset-B} of $\algmult$ is ignored.
\item In place of $V$, the ground set $W$ is used.
\end{itemize}
Our distributed algorithm is explicitly given in \algref{alg:distributed} and uses subroutine {\sc PartitionAndSample}, which is given in \algref{alg:partition}.

\begin{algorithm}[!htb]
\caption{{\sc PartitionAndSample}}
\alglab{alg:partition}
\begin{algorithmic}[1]
\Require{Set of elements $V$.}  
  \State $F \leftarrow \textrm{sample each $e\in V$ with probability $p = 4\sqrt{\memory/n}$}$
  \State Partition $V$ randomly into sets $V_1, V_2, \ldots V_\machines$ to the $\machines$ machines (one set per machine) 
  \State Send $F$ to each machine and a central machine $C$
\end{algorithmic}
\end{algorithm} 

\begin{algorithm}[!htb]
\caption{A 2-round distributed algorithm for ARMSM under knapsack constraints.}
\alglab{alg:distributed}
\begin{algorithmic}[1]
\Require{Parameters $d, m, K$, estimate $\tau^*$ of $f(\OPT)$.}  
\Statex{\textbf{Round 1:}\;}
\State{$F,V_1,\ldots,V_\machines \leftarrow \textsc{PartitionAndSample}(V)$}
\For{each machine $M_i$ (in parallel)}
\State $\tau \leftarrow f(\OPT)$
\State Let $\cB_0$ be the data structure of sets $B_{i, j}$ maintained by $\algmult_{\emptyset, F}(d, m, K, \tau)$
\If{$|\cB_0|<\memory$}
\State $R_i$ be the set of elements $e \in V_i$ that are in $\cB$ \label{alg:distributed-element-e}
\Else
\State $R_i \leftarrow \emptyset$
\EndIf
\State Send $R_i$ to a central machine $C$
 \EndFor
\vskip 0.1in\noindent
\textbf{Round 2 (only on $C$):}\;
\State{Compute $\cB_0$ from $F$ as in first round}
\State{$R \gets \cup_i R_i$} \\
\Return{$S_\tau \gets \algmult_{\cB_0, R}(d, m, K, \tau)$ \label{line:distributed-return-S}}
\label{alg:one-half-two-round-with-guess}
\end{algorithmic}
\end{algorithm} 

We formally prove \thmref{thm:distributed} in Appendix~\ref{app:missing:distributed} by first showing that the approximation guarantee is the same as \thmref{thm:mult}.
\begin{lemma}\lemlab{lemma:distributed-approximation}
There exists a distributed algorithm that outputs a set $Z$ so that $f(Z)$ has the same approximation guarantee as stated by \thmref{thm:mult}.
\end{lemma}
We can also bound the total number of elements sent to the central machine, using a proof similar to \cite{LiuV19}. 
\begin{lemma}\lemlab{lemma:distributed-memory-bound}
Let $\memory$ be an upper bound on the number of elements collected by \algmult. 
With probability $1 - e^{-\Omega(\memory)}$, the number of elements sent to the central machine $C$ is at most $\sqrt{n \memory}$.
\end{lemma}

\section{Experiments}
\label{sec:experiments}
In this section, we provide empirical evaluation of our algorithms for ARMSM under both single knapsack and multiple knapsack constraints. As no prior work exists in this setting we use the most natural generalizations of standard non-robust algorithms for comparison.
We test our most general algorithm \algmult against such algorithms while measuring the number of elements collected and the quality of the resulting approximation. 
The aim of our evaluations is to address the following points: 
\begin{enumerate}
\item
How does $\algmult{}$ compare to ``robustified'' generalizations of other submodular maximization algorithms?
\item
How well does $\algmult{}$ perform on real datasets compared to our theoretical worst-case guarantees?
\item
How many elements does \algmult collect?
\item
Does the performance of \algmult degrade as the number of elements $m$ removed at the end of the stream increases?
\end{enumerate}
\begin{sloppypar}
Implementation is available at \url{https://github.com/KDD2019SubmodularKnapsack/KDD2019SubmodularKnapsack}.
\end{sloppypar}
\paragraph{Robustification.} 
Although there are no existing ARMSM algorithms for knapsack constraints, we propose the following modification to existing algorithms to ensure a fair comparison.
Given a submodular maximization algorithm $\mathcal{A}$, we consider its \emph{robustified} version by allowing the algorithm to collect extra elements to obtain its own robust summary. 
To achieve this, we increase the knapsack capacity by some multiplicative factor, which is selected in such way that all algorithms collect approximately the same number of elements.

\subsection{Baselines}\label{sec:baselines} 
We compare \algmult to the following algorithms. 
\paragraph{Robustified \thresholding.} 
This algorithm corresponds to a robustified version of Algorithm~$2$ from~\cite{HuangKY17}, which accepts any element whose marginal density with respect to the stored elements exceeds a certain threshold. 
Note that while the algorithm is for a single knapsack constraint, it can be trivially extended to multiple knapsack constraints by checking that the thresholding condition holds for all dimensions. 
This marginal density thresholding algorithm is a natural generalization to knapsack constraints of the streaming algorithm \sieve{} \cite{BadanidiyuryMKK14} which gives the best theoretical guarantee under the cardinality constraint. 

\paragraph{Robustified offline \greedy.}  
This algorithm builds its summary by iteratively adding to it an element with the largest marginal density. 
Observe that \greedy is an offline algorithm, which is a more powerful model. 
However, \greedy is a single knapsack algorithm, so we use it only as a baseline for single knapsack constraints.
While there exists a \greedy algorithm~\cite{LinB10} under multiple knapsack constraints, it requires $\O{n^5}$ running time, which makes it infeasible on large datasets.
\paragraph{Robustified \multidim.}
This is a robustified version of the streaming algorithm for submodular maximization with multiple knapsack constraints from~\cite{YuXC18}. 

\subsection{Objectives and Datasets}
We evaluate the algorithms on two submodular objective functions:
\paragraph{Dominating set.} 
We use graphs \texttt{ego-Facebook} (4K vertices, 81K edges) and \texttt{ego-Twitter} (88K vertices, 1.8M edges) from the SNAP database~\cite{snapnets}. 
For a graph $G(V,E)$ and $Z \subseteq V$, we let $f(Z) = \frac{|Z \cup N(Z)|} {|V|}$, where $N(Z)$ is the set of all neighbors of $Z$. 
For each knapsack constraint, the cost of each element is selected uniformly at random from the uniform distribution $\mathcal{U}(1, 3)$ and all knapsack constraints are set to $10$. 

\paragraph{Movie recommendation.}
Modeling the scenario of movie recommendations we analyze a dataset of movie ratings (in the range $[1, 5]$) assigned by users. 
For each movie $x$ we use a vector $v_x$ of normalized ratings: if user $u$ did not rate movie $x$, then set $v_{x, u} = 0$, otherwise set $v_{x, u} = r_{x, u} - r_{avg}$, where $r_{avg}$ denotes the average of all known ratings. 
Then, the similarity between two movies $x_1$ and $x_2$ can be defined as the dot product $\langle v_{x_1}, v_{x_2} \rangle$ of their vectors.

In the case of movie recommendation the goal is to select a representative subset of movies.
The domain of our objective is the set of all movies. 
For a subset of movies $X$ we consider a parameterized objective function $f_X$: 
$$f_X(Z) = \sum_{x \in X} \max_{z \in Z} \langle v_z, v_x \rangle,$$
where $Z$ is a subset of movies. This captures how representative is $Z$ of the set $X$.
In our experiments, we model the situation of making recommendations to some user so we pick $X$ to be a set of movies rated by the user (we select the user uniformly at random). Hence the maximizer of $f_X(Z)$ corresponds to a subset of movies which represents well user's rated set $X$.

We use the \texttt{ml-20} MovieLens dataset~\cite{movielens}, containing $27\,278$ movies and $20\,000\,263$ ratings.
Knapsack constraints model limited demand for movies of a certain type (e.g. not too many action movies, not too many fantasy movies, etc). 
In the data each movie is labeled by several genres and each knapsack constraint is described by sets of ``good'' and ``bad'' genres. 
Movies with more ``good'' genres and less ``bad'' genres have lower cost, allowing the algorithm to choose more such movies. If there are at most $t$ genres describing ``good'' and ``bad'' sets then we set the cost of a movie $x$ to be linear in the range $[1, t + 1]$:
\[ c(x) = 1 + 0.5 \times (bad(x)- good(x) + t), \]
where $good(x)$ and $bad(x)$ are the numbers of good and bad genres that movie $x$ is labeled with.

For the experiments under one knapsack constraint, we define the good set of movies as $good = \set{\text{Comedy},\ \text{Horror}}$, and the bad set of movies as $bad = \set{\text{Adventure},\ \text{Action}}$. 
For experiments under two knapsack constraints, we define the second constraint by an additional set of good movies $good = \set{\text{Drama},\ \text{Romance}}$, and an additional set of bad movies $bad = \set{\text{Sci-Fi},\ \text{Fantasy}}$. 
All knapsack constraint bounds are set to $10$, limiting the total number of recommended movies. 

\subsection{Experimental Evaluation and Results}
We compare \algmult against the three baselines described in Section~\ref{sec:baselines}. 
First, we obtain robust summaries for \algmult and for each of the baselines. 
Second, we adversarially remove elements from these summaries.
Finally, we run \offline on the remaining elements in the summaries and compare the values of objective functions on the resulting sets.

\paragraph{Adversarial removals.} 
To ensure a fair comparison, we use the same set of removed elements for all algorithms. 
This is done by removing the union of sets recommended by all algorithms and then continuing in a recursive fashion if more removals are required. 

We define the removal process formally as follows.
For an algorithm $\cA$, let $S_{\cA}$ be the robust summary output by $\cA$. 
We let $R_1 = \cup_{\cA} \offline{(S_{\mathcal{A}})}$, where the union is taken over all four algorithms $\cA$ tested.
That is, $R_1$ is the union of the \emph{best} elements selected using \algmult, \greedy, \multidim, and \thresholding.
This typically already gives a good choice of removals. 
If more removals are required, we define $R_{k+1} = \cup_{\cA} \offline(S_{\cA} \setminus \cup_{i=1}^k R_i)$. 
That is, we recursively remove the union of the elements in the optimal sets across all the algorithms and we repeat this process until $R_k$ is empty. 

\begin{figure*}[!ht]
	\centering
		\subfloat[\texttt{ml-20}, $1$ knapsack]{\includegraphics[width=0.33\textwidth]{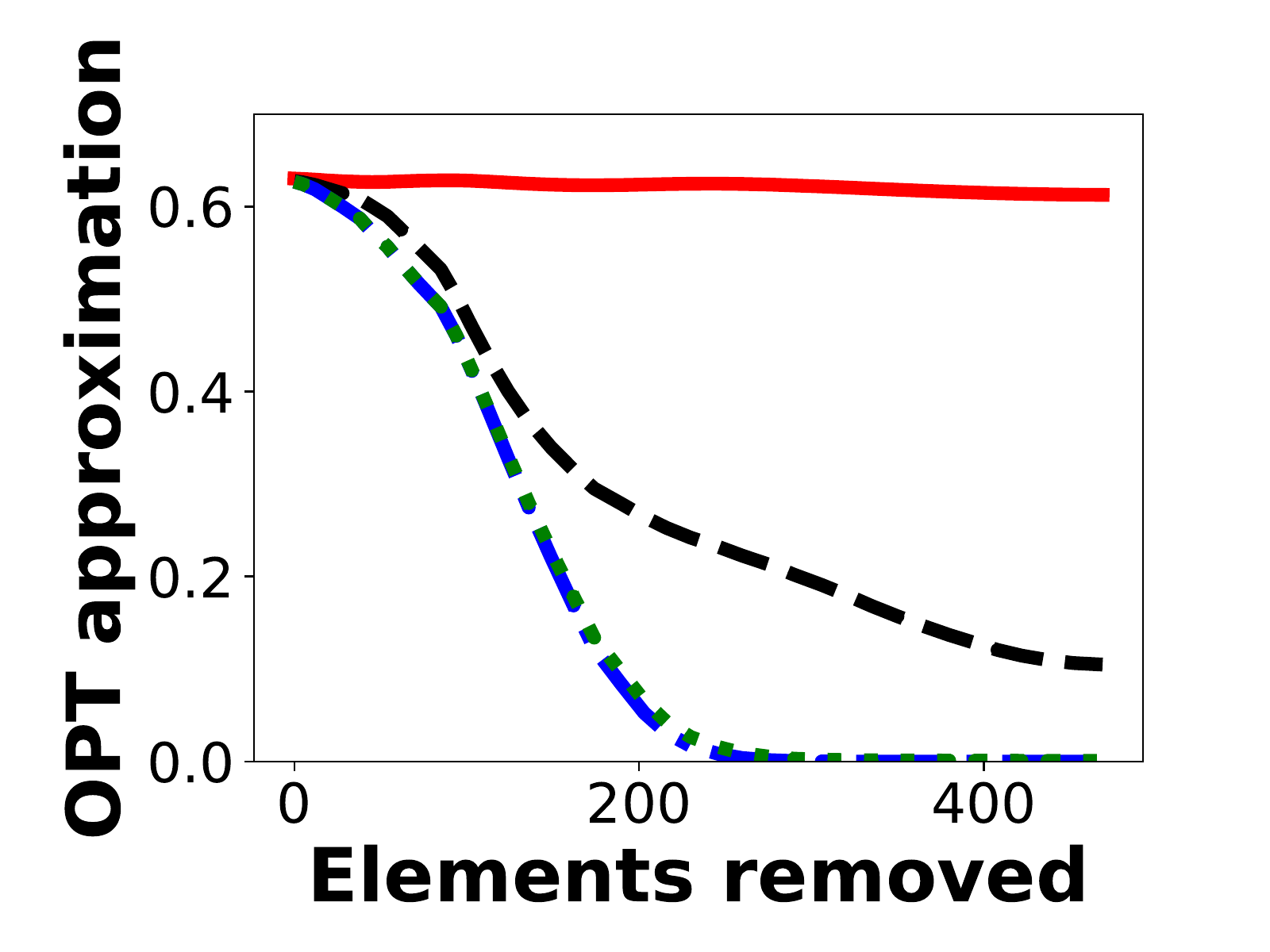} \label{fig:res1}}
		\subfloat[\texttt{ego-Facebook}, $1$ knapsack]{\includegraphics[width=0.33\textwidth]{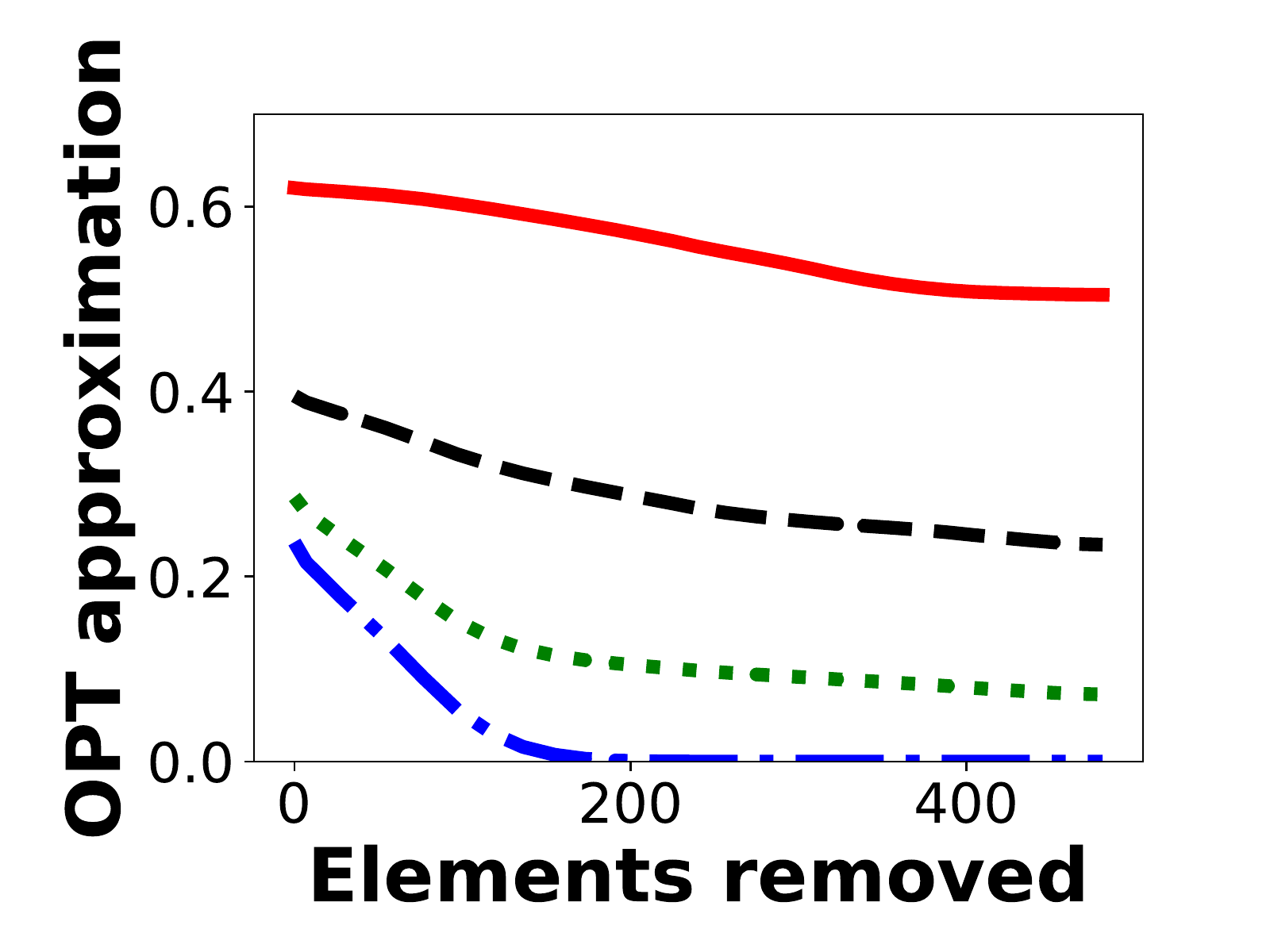} \label{fig:res2}}
		\subfloat[\texttt{ego-Twitter}, $1$ knapsack]{\includegraphics[width=0.33\textwidth]{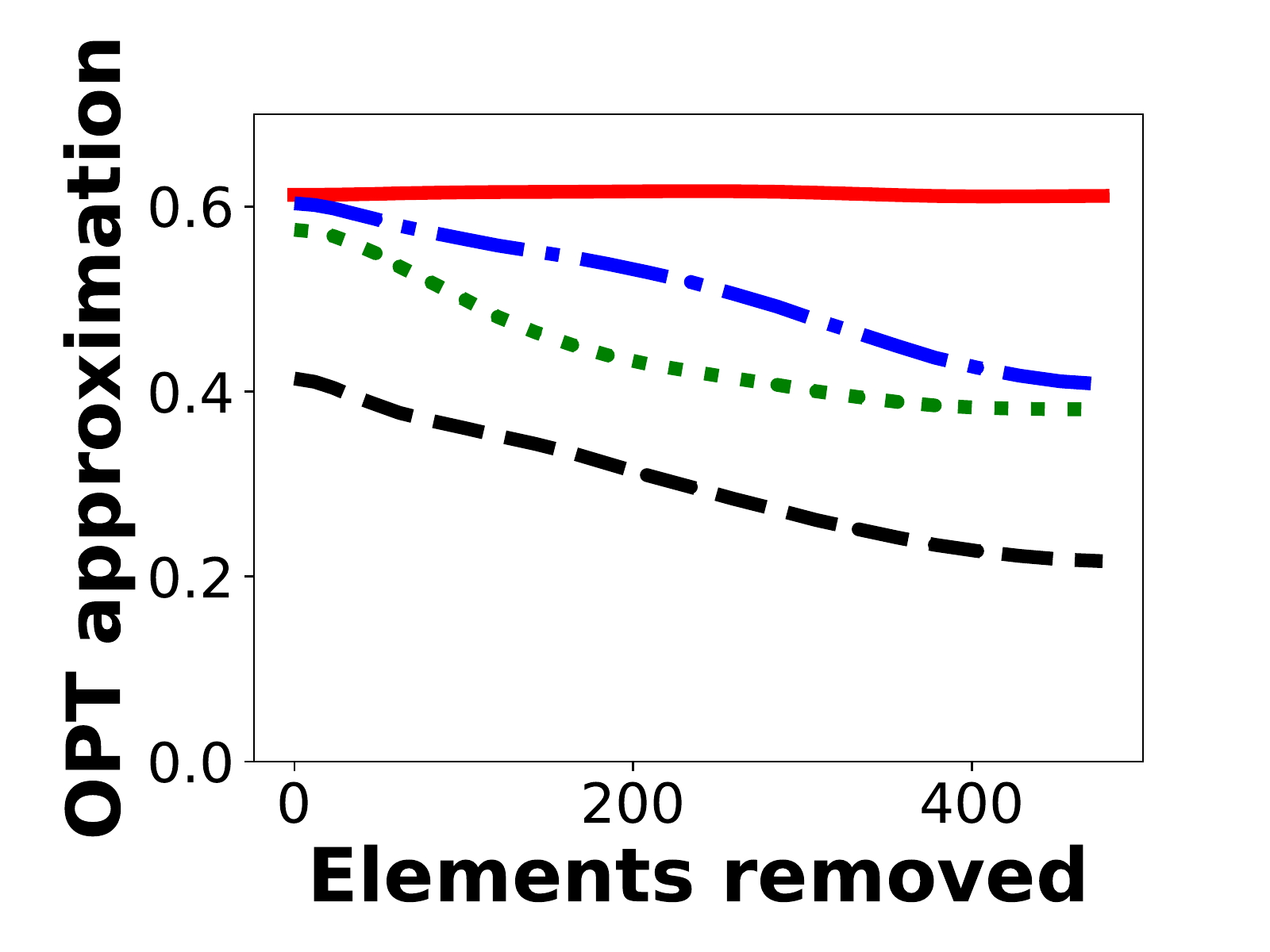} \label{fig:res3}}
	\\
	%\vspace{-0.78cm}
		\subfloat[\texttt{ml-20}, $2$ knapsacks]{\includegraphics[width=0.33\textwidth]{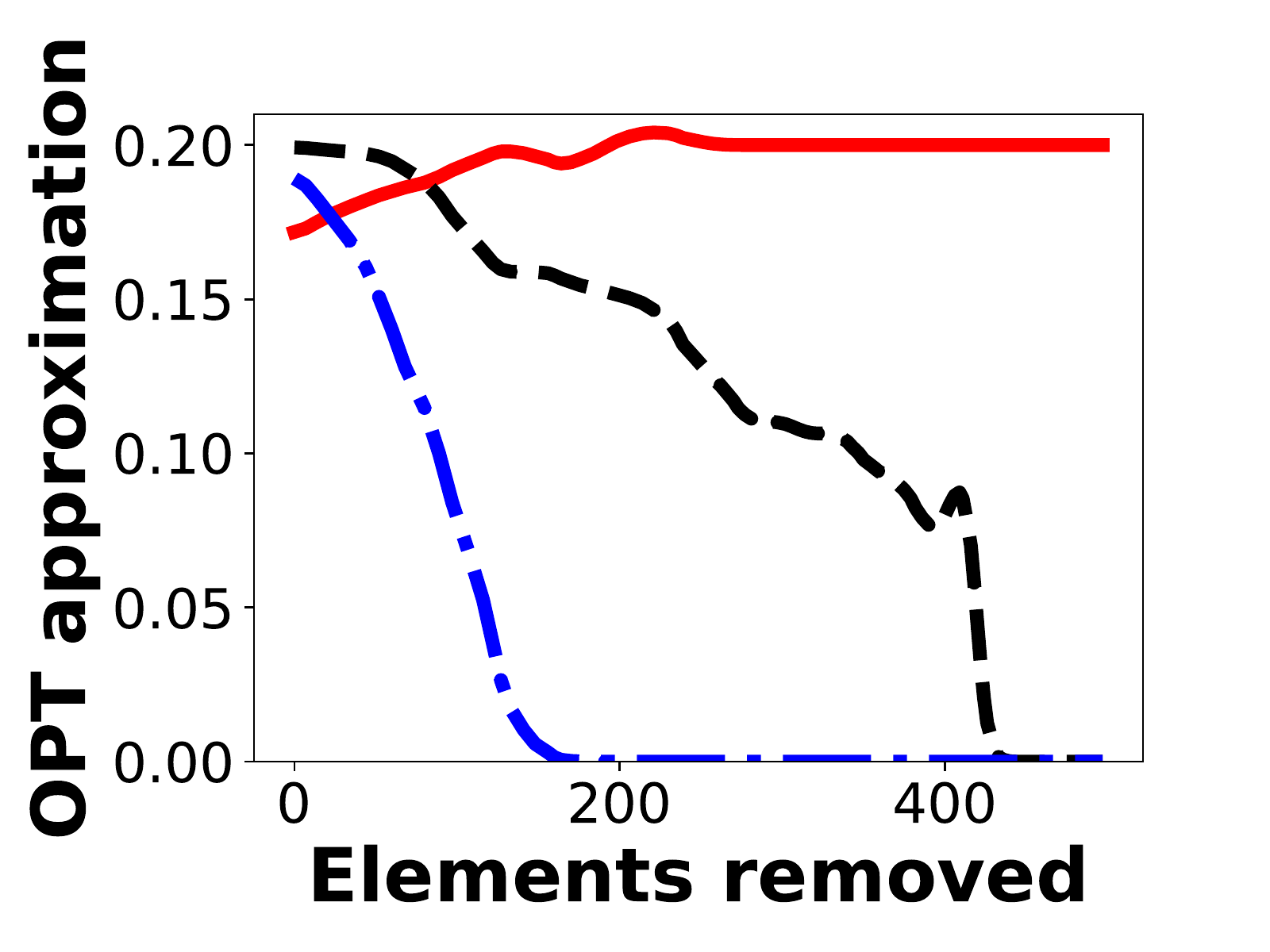}\label{fig:res4}}
		\subfloat[\texttt{ego-Facebook}, $2$ knapsacks]{\includegraphics[width=0.33\textwidth]{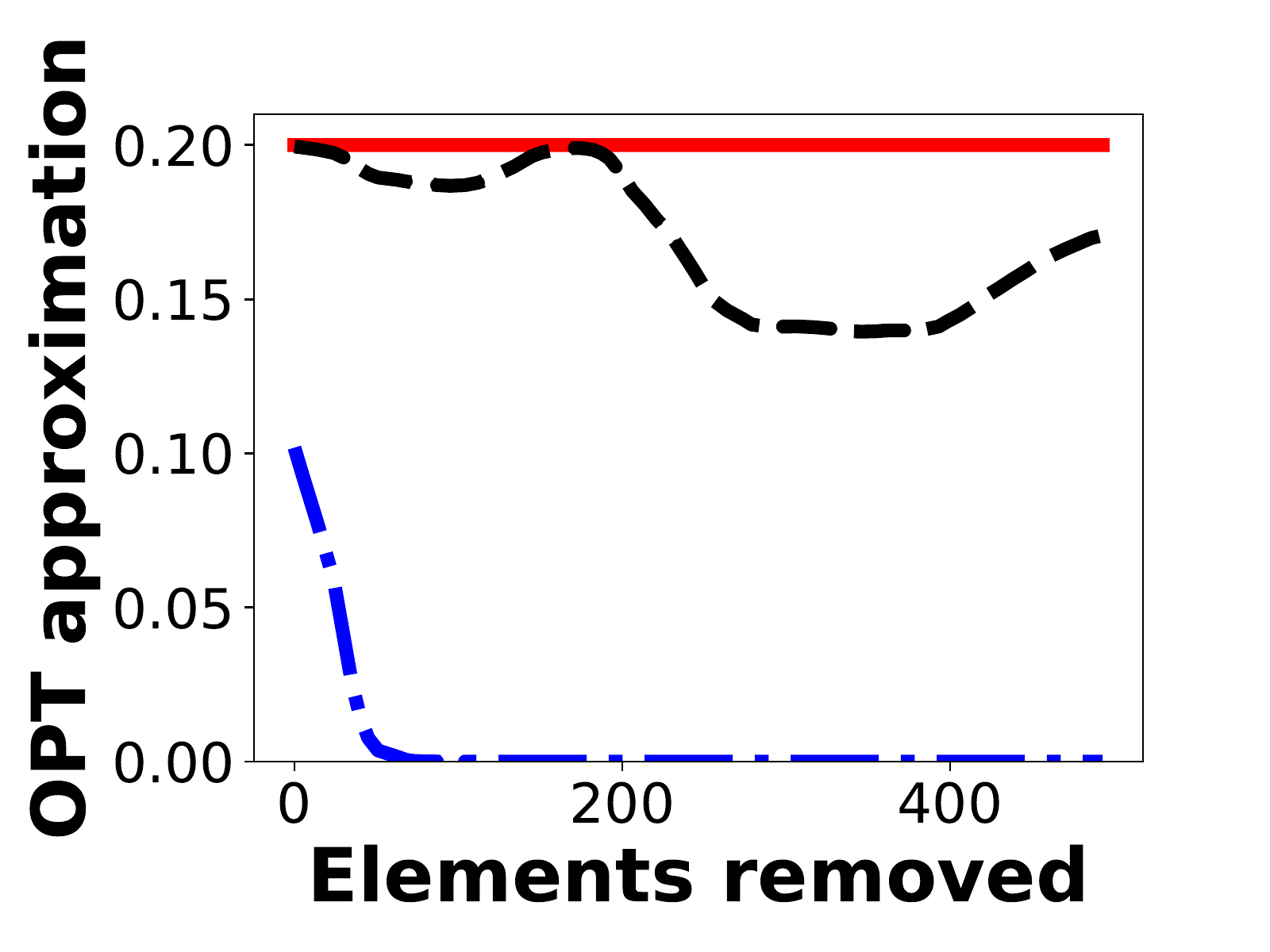}\label{fig:res5}}
		\subfloat[\texttt{ego-Twitter}, $2$ knapsacks]{\includegraphics[width=0.33\textwidth]{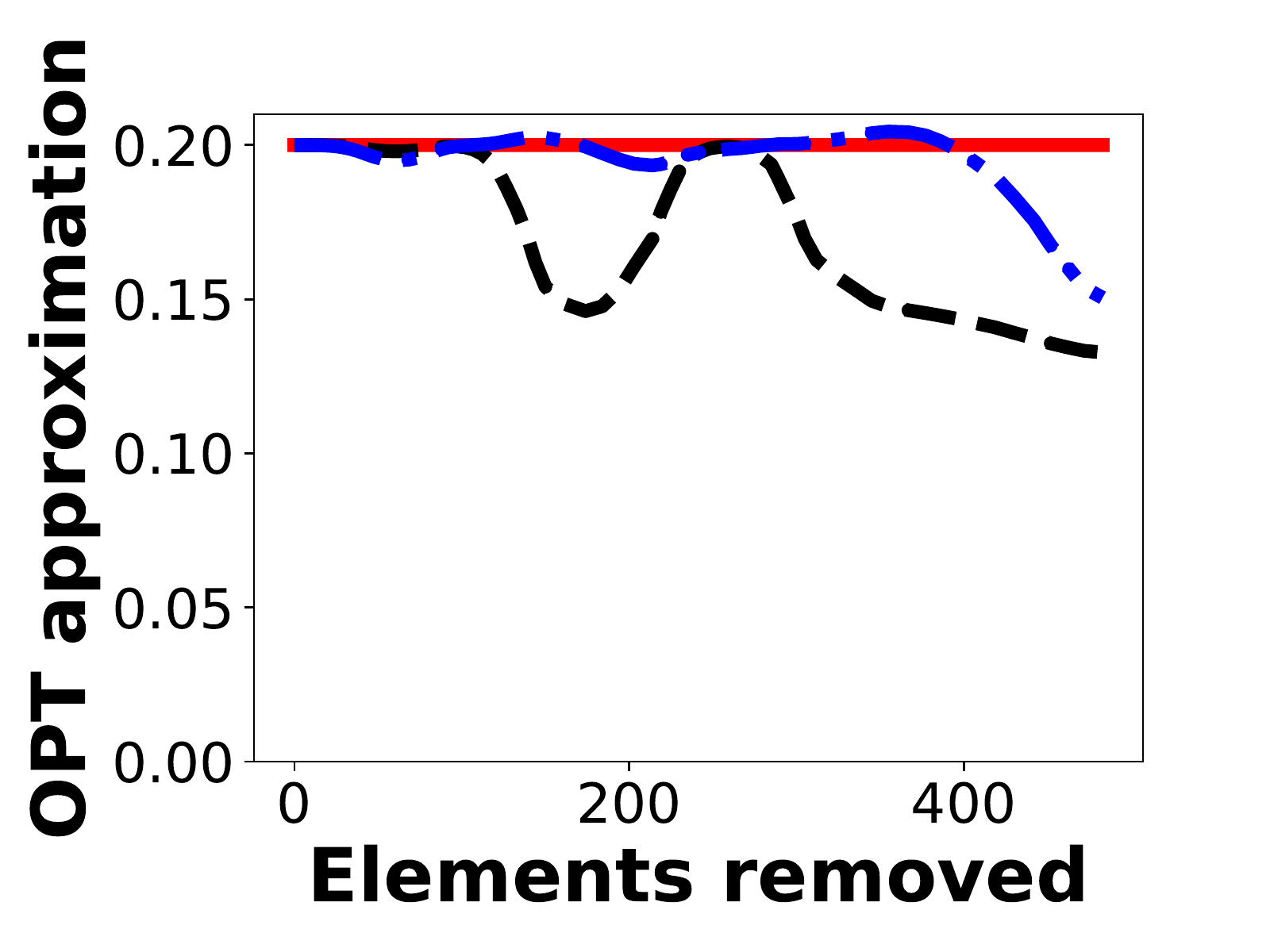}\label{fig:res6}}
	\\
	%\vspace{-0.78cm}
	\subfloat{\includegraphics[width=\textwidth]{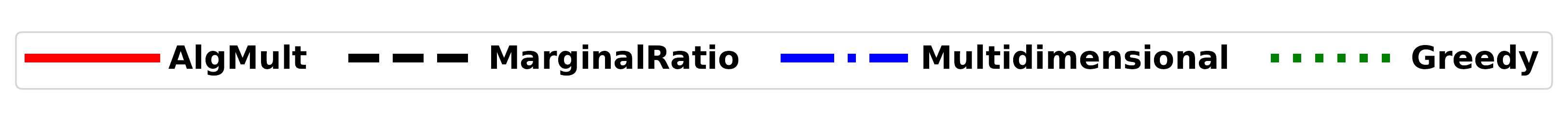}}
	\vspace{-0.4cm}
	\caption{Approximation of $f(\OPT)$ for different algorithms (K = 10).
	%For each algorithm we divide the objective return by the algorithm by the objective returned by corresponding \offline which knows all removals in advance and multiply by \offline's approximation factor.
	In most cases, \algmult outperforms the other algorithms and achieves the best possible approximation factor ($1 - \frac 1 e \approx 0.63$ for one knapsack, $0.2$ for two knapsacks).}
	\label{fig:results}
\end{figure*}
\begin{table*}[!htb]
	\centering
	\tiny
	\begin{tabular}[t]{*7c}
	\hline
& ml-20, 1 knapsack & fb, 1 knapsack & twitter, 1 knapsack & ml-20, 2 knapsacks & fb, 2 knapsacks & twitter, 2 knapsacks \\ \hline
	\algmult                           & 641      & 378   & 401        & 1350     & 2745  & 4208 \\ \hline
	\thresholding & 641      & 377   & 402        & 1350     & 2745  & 4209 \\ \hline
	\multidim          & 87       & 18    & 435        & 72       & 22    & 4221 \\ \hline
	\greedy                            & 647      & 393   & 493        & -        & -     & - \\ \hline
	\end{tabular}
	\caption{Sizes of robust summaries produced by the algorithms ($K = 10$).}
	\label{tab:elements}
\end{table*}

%\begin{figure*}[!ht]
	%\centering
		%\subfloat[\texttt{ml-20}, $1$ knapsack]{\includegraphics[width=0.33\textwidth]{pics/large_rem.pdf} \label{fig:res1}}
		%\subfloat[\texttt{ego-Facebook}, $1$ knapsack]{\includegraphics[width=0.33\textwidth]{pics/fbUniform_rem.pdf} \label{fig:res2}}
		%\subfloat[\texttt{ego-Twitter}, $1$ knapsack]{\includegraphics[width=0.33\textwidth]{pics/twitterUniform_rem.pdf} \label{fig:res3}}
	%\\
	%%\vspace{-0.78cm}
		%\subfloat[\texttt{ml-20}, $2$ knapsacks]{\includegraphics[width=0.33\textwidth]{pics/largeMD2_rem.pdf}\label{fig:res4}}
		%\subfloat[\texttt{ego-Facebook}, $2$ knapsacks]{\includegraphics[width=0.33\textwidth]{pics/fbUniformMD_rem.pdf}\label{fig:res5}}
		%\subfloat[\texttt{ego-Twitter}, $2$ knapsacks]{\includegraphics[width=0.33\textwidth]{pics/twitterUniformMD2_rem.pdf}\label{fig:res6}}
	%\\
	%%\vspace{-0.78cm}
	%\subfloat{\includegraphics[width=\textwidth]{pics/legend.pdf}}
	%\vspace{-0.4cm}
	%\caption{Approximation of $f(\OPT)$ for different algorithms (K = 10).
	%%For each algorithm we divide the objective return by the algorithm by the objective returned by corresponding \offline which knows all removals in advance and multiply by \offline's approximation factor.
	%In most cases, \algmult outperforms the other algorithms and achieves the best possible approximation factor ($1 - \frac 1 e \approx 0.63$ for one knapsack, $0.2$ for two knapsacks).}
	%\label{fig:results}
%\end{figure*}
%\input{tables/elements.tex}

\paragraph{Evaluation.} 
For different numbers of removed elements, we compare the values that are produced by the offline algorithm on robust summaries, i.e. $\offline{(S_{\cA} \setminus \bigcup_{i = 1}^k{R_i})}$ generated by the four algorithms. 
Since $f(\OPT)$ is NP-hard to compute, we compare the performance of each algorithm with upper bounds on $f(\OPT)$ to estimate the approximation given by the algorithms. 
For a single knapsack constraint, the best known upper bound can be computed from \greedy and for multiple knapsack constraints from \multidim~\cite{YuXC18}.

\paragraph{Results.} 
The results of our experiments are shown in Figure~\ref{fig:results}. 
For each algorithm, we plot the ratio of its objective to an upper bound on the optimal solution, which is obtained as previously discussed. 
Figures~\ref{fig:res1},~\ref{fig:res2}, and~\ref{fig:res3} show experimental results for $1$-knapsack constraints using \greedy as the offline algorithm with approximation factor $1 - e^{-c(s)/K}$, where $K$ is the knapsack constraint and $c(S)$ is the cost of the resulting set. 
For many instances, $c(S)$ is close to $K$, so this value is close to $1 - \frac 1 e \approx 0.63$. 
Figures~\ref{fig:res4},~\ref{fig:res5}, and~\ref{fig:res6} show experimental results for $2$-knapsack constraints.

Our evaluations suggest that \algmult provides the best possible approximation factor for a majority of inputs. 
Except for the first iterations in Figure~\ref{fig:res4}, \algmult outperforms the other algorithms, and achieves roughly the same approximation guarantee as the offline algorithm that knows the items to be removed in advance. 
In fact, the advantage of \algmult becomes more noticeable as larger numbers of elements are removed. 

Since the baseline algorithms, other than \greedy, require an estimate of $f(\OPT)$, we try several such estimations.
The non-monotone behavior of the ratio of \thresholding to $f(\OPT)$ in Figure~\ref{fig:res6} occurs since \thresholding performs better when estimation is close to the true objective. 
We emphasize the fact that all algorithms, including \algmult, use the same $f(\OPT)$ estimations. 
It is possible to obtain a more monotone behavior by trying more estimations, but doing so will require collecting more elements. 

To evaluate memory consumption, we also report the number of elements collected by each algorithm. 
These results are presented in Table~\ref{tab:elements} and show that the algorithms for $2$-knapsack constraints collect noticeably more elements than those performing maximization under a single knapsack constraint. 
The size of the robust summary output by \algmult does not appear to correlate with the total number of elements, and in the case of \texttt{ego-Twitter}, it collects only $5\%$ of the vertices. 

Recall that we allow the baseline algorithms to collect extra elements by increasing their knapsack capacity to ensure fair comparison. 
Hence, almost all the algorithms collect similar numbers of elements for each setup, as shown in Table~\ref{tab:elements}.
Note that, however, for some experimental setups \multidim collects significantly fewer elements than the other algorithms. This phenomenon persists even if the knapsack capacity is unbounded.

In our empirical evaluations, the number of collected elements did not seem to depend on the number of removed items $m$. 
One possible reason for this phenomena is that the algorithms were not executed with small guesses for the optimal objective. 
As a result when the number of removed elements is large, the optimal objective is below the threshold considered by the algorithm, and therefore more elements are not collected because the threshold is set to be too high. 
However, it is natural that with sufficiently bad guesses for the optimal objective, any thresholding algorithm will be forced to meaninglessly collect a large number of elements. 
\section{Conclusion}
We have given the first streaming and distributed algorithms for adversarially robust monotone submodular maximization subject to single and multiple knapsack constraints. 
Our algorithms are based on a novel data structure which dynamically allocates new space depending on the elements stored so far and perform well on large scale data sets, even compared to offline algorithms that know in advance which elements will be removed. 

For the future work, it is natural to ask whether our framework can be scaled to larger datasets for some specific classes of objectives, e.g., is it possible to ensure adversarial robustness with sketching methods for coverage objectives~\cite{BateniEM18}? 
It would be also interesting to understand the limits on approximation that can be achieved with adversarial robustness and summary size only $\Tilde O(K + m)$.
Finally, an interesting open question is whether it is possible to do adversarially robust \emph{non-monotone} submodular maximization. 

\bibliography{references}
\bibliographystyle{plainurl}
\appendix
\section{Missing Proofs from Section~\ref{sec:single}}
\label{app:missing:single}
We call a bucket \emph{saturated} if the cost of its items is at least half of its capacity. 
\begin{definition}
A bucket $B_{i,j}$, is \emph{saturated} if $c(B_{i,j})\ge 2^i$.
\end{definition}
We break the analysis into the following three cases:
\begin{compactenum}
\item
At least half of the buckets in some partition are saturated (\lemref{lem:items:saturated})
\item
More than half of the buckets in all partitions are not saturated, but there exists some bucket in the last partition $B_\numpartitions$ which is a good estimate of $f(\OPT)$ (\lemref{lem:items:good})
\item
More than half of the buckets in all partitions are not saturated and no bucket in the last partition $B_\numpartitions$ is a good estimate of $f(\OPT)$ (\lemref{lem:items:bad})
\end{compactenum}

We first show that if at least half of the buckets in some partition are saturated, some saturated bucket in this partition cannot be affected too much by the removal of elements $E$ at the end of the stream. 
Hence, this saturated bucket contains a set of elements $Z$ such that $f(Z)$ is a good approximation to $\tau$, which in turn is a good approximation to $f(\OPT)$. 

Let $S_\tau := \{B_{i,j}\}_{i,j}$ denote the data structure output by \algnum when run with parameter $\tau$  (i.e. for $\tau^* = \frac{32 (1 - \frac{1}{2 \numpartitions}) + 3}{2}\tau$). 
Let $Z_\tau := \offline(\bigcup_{B_{i,j} \in S_\tau} B_{i,j}\setminus E)$. 
\begin{lemma}
\lemlab{lem:items:saturated}
Let $\numpartitions=\ceil{\log K}$ and $\tau>0$.  
If there exists a partition $B_{i^*}$ in $S_\tau$ with at least half of its buckets saturated, then for the set $Z_\tau$ it holds that:
\[f(Z_\tau)\ge\left(1-\frac{1}{e}\right)\left(1-\frac{1}{2\numpartitions}\right)\tau.\]
\end{lemma}
\begin{proof}
Let $B_{i^*}$ be a partition in $S_\tau$ with at least half of its buckets saturated. 
Let $B_{i^*,j^*}$ be a saturated bucket that minimizes the cost of removed items, i.e. $c(B_{i^*,j}\cap E)$, among all saturated buckets in this partition. 
Let $I$ be the cost all items in $E$ which are in partition $i^*$, i.e. $I=c\left(E \cap \bigcup_{j=1}^{\numbuckets_i}B_{i^*,j}\right)$.
Then the total capacity of all buckets in partition $i^*$ is at least $2^{i^*+1}\cdot w\ceil{K/2^{i^*}}+8\numpartitions I$. 
Thus, the total number of buckets in partition $i^*$ is at least $\frac{2^{i^*+1}\cdot w\ceil{K/2^{i^*}}+8\numpartitions I}{2^{i^*}}$. 
Since at least half of its buckets are saturated the total number of the saturated buckets in partition $i^*$ is at least $\frac{2^{i^*}\cdot w\ceil{K/2^{i^*}}+4\numpartitions I}{2^{i*}}$. %for $\numpartitions=\ceil{\log 2K}$ and $\capacity_i=\min\{2^i,K\}$, 
By an averaging argument: 
$$c(B_{i^*,j^*}\cap E)\le\frac{2^{i^*}\cdot I}{\capacity_{i*}\cdot w\ceil{K/2^{i^*}}+4\numpartitions I} < \frac{2^{i^*}}{4 \numpartitions}.$$
Thus, $c(B_{i^*,j} \setminus E)\ge c(B_{i^*,j^*})- \frac{\capacity_{i*}}{4 \numpartitions}$. 
Since $B_{i^*,j^*}$ is saturated by definition, then $2^{i^*-1}\le c(B_{i^*,j^*}) \le 2^{i^*}$, and hence the marginal density of each element exceeds a threshold of $\frac{\tau}{2^{i^*}} \ge \frac{\tau}{c(B_{i^*,j^*})}$ so that: 
$$f(B_{i^*,j^*} \setminus E)\ge\left(c(B_{i^*,j^*})- \frac{\capacity_{i*}}{4 \numpartitions}\right)\frac{\tau}{c(B_{i^*,j^*})} \ge \left(1 - \frac{1}{2\numpartitions}\right) \tau.$$ 
%Using the observation that $\frac{I}{\capacity_i\cdot w\ceil{K/2^i}+(4\numpartitions)I}<\frac{1}{4\numpartitions}$ and $c(B_{i^*,j})\ge 2^{i^*-1}\ge 1$ for a saturated bucket $B_{i^*,j}$, then it follows that $f(B_{i^*,j} \setminus E)\ge\left(1-\frac{1}{4\numpartitions}\right)\tau$
Hence by \thmref{thm:offline}, running \offline{} on $B_{i^*,j^*} \setminus E$ gives value at least $\left(1-\frac{1}{e}\right)\left(1-\frac{1}{2\numpartitions}\right) \tau$.
\end{proof}

Before considering the other cases, we need the following technical lemmas. 
\lemref{lem:Ei} bounds the value of the removed elements, while \lemref{lem:recursion} allows us to relate the elements removed in a bucket $B_{\numpartitions, r}$ of the last partition with the elements removed in previous partitions. 
\begin{lemma}
\lemlab{lem:Ei}
Given a bucket $A_{i-1}$ from partition $i-1$ that is not saturated, then the loss in bucket $A_i$ induced by the removals is at most
\[\margain{E_i}{A_{i-1}}<\frac{\tau}{2^{i-1}}c(E_i),\]
where $E_i:=A_i\cap E$ denotes the elements that are removed from $A_i$. 
\end{lemma}
\begin{proof}
By submodularity,
\begin{equation}
\eqnlab{eqn:subsum}\margain{E_i}{A_{i-1}}\le\sum_{e\in E_i}\margain{e}{A_{i-1}}.
\end{equation}
For each $e\in E_i$, either $\frac{f(e)}{c(e)}<\frac{\tau}{2^{i-1}}$ or $\frac{f(e)}{c(e)}>\frac{\tau}{2^{i-1}}$. 
In the first case, $\margain{e}{A_{i-1}}\le f(e)\le\frac{\tau}{2^{i-1}}\cdot c(e)$. 
In the second case, since $e\in A_i$ it must hold that $c(e)\le 2^{i-1}$. 
On the other hand, $e\notin A_{i-1}$ but $c(A_{i-1})<2^{i-1}$ because $A_{i-1}$ is not saturated. 
Thus, $\margain{e}{A_{i-1}}<\frac{\tau}{2^{i-1}}\cdot c(e)$ or else the algorithm would have added $e$ to $A_{i-1}$.
Hence, $\margain{e}{A_{i-1}}<\frac{\tau}{2^{i-1}}\cdot c(e)$ for all $e\in E_i$ and so by \eqnref{eqn:subsum}, $\margain{E_i}{A_{i-1}}<\frac{\tau}{2^{i-1}}c(E_i)$.
\end{proof}

\begin{lemma}
\lemlab{lem:recursion}
Suppose that there exists some bucket in every partition that is not saturated. 
Let $\numpartitions=\ceil{\log K}$. 
For every partition $i$, let $A_i$ denote a bucket with $c(A_i)<2^i$ and let $E_i:=A_i \cap E$ denote the elements that are removed from $A_i$. 
The loss in the bucket $B_{\numpartitions, \lastbucket}$ induced by the removals, given the remaining elements in the previous buckets, is at most 
\[\margain{E_{\numpartitions}}{\bigcup_{j = 0}^{\numpartitions - 1}\left(A_j \setminus E_j \right)} \le \sum_{j = 1}^{\numpartitions} \frac{\tau}{2^{j - 1}} c(E_j).
	\]
\end{lemma}
\begin{proof}
We show by induction that for any $i \ge 1$ the following holds
\begin{equation}
\eqnlab{eqn:inductive-inequality}
\margain{E_i}{\bigcup_{j = 0}^{i - 1}\left(A_j \setminus E_j \right)} \le \sum_{j = 1}^{i} \frac{\tau}{2^{j - 1}} c(E_j).
\end{equation}
so that the claim by setting by setting $i = \numpartitions$.	
\paragraph{Base case $i = 1$.}
Since $c(A_0)<2^0=1$ and each item is normalized to have cost at least $1$, it follows that both $A_0$ and $E_0$ are empty. 
Thus
\[f(E_1)\le\frac{\tau}{2^0}=\sum_{j=1}^1\frac{\tau}{2^{j-1}}c(E_j),\]
where the first inequality holds by \lemref{lem:Ei}.

\paragraph{Inductive step $i > 1$.}	
Assuming \eqnref{eqn:inductive-inequality} holds for $i-1$ where $i>1$, we now show that it also holds for $i$. 
By submodularity, $\margain{E_{i - 1}}{\bigcup_{j = 0}^{i - 2} \left(A_j \setminus E_j\right)} \geq \margain{E_{i - 1}}{\bigcup_{j = 0}^{i - 1} \left(A_j \setminus E_j\right)}$. 
It follows that $\margain{E_i}{\bigcup_{j = 0}^{i - 1} \left(A_j \setminus E_j\right)} \le \margain{E_i}{\bigcup_{j = 0}^{i - 1} \left(A_j \setminus E_j\right)} + \margain{E_{i - 1}}{\bigcup_{j = 0}^{i - 2} \left(A_j \setminus E_j\right)} - \margain{E_{i - 1}}{\bigcup_{j = 0}^{i - 1} \left(A_j \setminus E_j\right)}$, by adding $\margain{E_i}{\bigcup_{j = 0}^{i - 1} \left(A_j \setminus E_j\right)}$ to both sides.
Since $E_i$ and $\bigcup_{j = 0}^{i - 1} \left(A_j \setminus E_j\right)$ are disjoint, then 
\begin{align}
\eqnlab{eqn:Ei:given-union}  
\margain{E_i}{\bigcup_{j = 0}^{i - 1} \left(A_j \setminus E_j\right)} \le f \left(E_i \cup  \bigcup_{j = 0}^{i - 1} \left(A_j \setminus E_j\right)\right) \notag\\
+ \margain{E_{i - 1}}{\bigcup_{j = 0}^{i - 2} \left(A_j \setminus E_j\right)} - f \left(E_{i - 1} \cup \bigcup_{j = 0}^{i - 1} \left(A_j \setminus E_j\right)\right). 
\end{align}

We can bound the first term $f \left(E_i \cup  \bigcup_{j = 0}^{i - 1} \left(A_j \setminus E_j\right)\right)$ by at most $f \left(E_i \cup A_{i - 1} \cup \bigcup_{j = 0}^{i - 2} \left(A_j \setminus E_j\right)\right)$, by monotonicity.
Note that the third term $f \left(E_{i - 1} \cup \bigcup_{j = 0}^{i - 1} \left(A_j \setminus E_j\right)\right)$ equals $f \left(E_{i - 1} \cup A_{i - 1} \cup \bigcup_{j = 0}^{i - 2} \left(A_j \setminus E_j\right)\right)$, which is at least $f \left(A_{i - 1} \cup \bigcup_{j = 0}^{i - 2} \left(A_j \setminus E_j\right)\right)$, because $E_{i - 1} \cup \left(A_{i - 1} \setminus E_{i - 1} \right) = E_{i - 1} \cup A_{i - 1}$. 

Hence, \eqnref{eqn:Ei:given-union} gives $\margain{E_i}{\bigcup_{j = 0}^{i - 1} \left(A_j \setminus E_j\right)}$ is at most the sum $\margain{E_i}{A_{i - 1} \cup \bigcup_{j = 0}^{i - 2} \left(A_j \setminus E_j\right)} + \margain{E_{i - 1}}{\bigcup_{j = 0}^{i - 2} \left(A_j \setminus E_j\right)}$. 
Then by submodularity, $\margain{E_i}{\bigcup_{j = 0}^{i - 1} \left(A_j \setminus E_j\right)}$ is at most the sum $\margain{E_i}{A_{i - 1}} + \margain{E_{i - 1}}{\bigcup_{j = 0}^{i - 2} \left(A_j \setminus E_j\right)}$. 
Since $c(A_{i-1})<2^{i-1}$, \lemref{lem:Ei} implies $\margain{E_i}{\bigcup_{j = 0}^{i - 1} \left(A_j \setminus E_j\right)} \le  \frac{\tau}{2^{i - 1}}c(E_i) + \margain{E_{i - 1}}{\bigcup_{j = 0}^{i - 2} \left(A_j \setminus E_j\right)}$.
By the inductive hypothesis, $\margain{E_i}{\bigcup_{j = 0}^{i - 1} \left(A_j \setminus E_j\right)} \le  \frac{\tau}{2^{i - 1}}c(E_i) + \sum_{j = 1}^{i - 1} \frac{\tau}{2^{j - 1}}c(E_j)  = \sum_{j = 1}^{i} \frac{\tau}{2^{j - 1}}c(E_j)$.
\end{proof}

We require the following structural lemma from \cite{MitrovicBNTC17}.
\begin{lemma}
\lemlab{lemma:A-B-R}
\cite{MitrovicBNTC17}
For any monotone, non-negative submodular function $f$ on a ground set $V$, and any sets $A,B,R \subseteq V$, we have
\[f(A \cup B) - f(A \cup (B \setminus R)) \le \margain{R}{A}.\]
 \end{lemma}

We also require the following structural lemma relating sets of large size to a large number of sets of small size. 
\begin{lemma}
\lemlab{lem:numsets}
Let $S$ be a set of size $\alpha K$ for some integer $\alpha\ge1$ such that no item in $S$ has size more than $K$. 
Then the items of $S$ can be partitioned into $2\alpha - 1$ sets, each with size at most $K$.
\end{lemma}
\begin{proof}
Let $S_1,\ldots,S_i$ be sets that partition $S$ with the minimal cardinality $i$. 
Note that $c(S_1)+c(S_2)>K$, or else the elements of set $S_1$ and $S_2$ can be combined into a single set, contradicting the definition of $i$. 
Similarly, $c(S_{2j-1})+c(S_{2j})>K$ for each integer $j>0$. 
On the other hand, $c(S)=\alpha K$, so $j<\alpha$ and there can be at most $2\alpha-1$ sets. 
\end{proof}

We can finally consider the second case of our analysis, where more than half of the buckets in all partitions are not saturated, but there exists some bucket $B_{\numpartitions, \lastbucket}$ in the last partition such that $f(B_{\numpartitions, \lastbucket})$ is a good estimate of $f(\OPT)$. 
\begin{lemma}
\lemlab{lem:items:good}
Let $\numpartitions=\ceil{\log K}$ and $\tau>0$.  
If no partition in $S_\tau$ has at least half of its buckets saturated, then:
\[f(Z_\tau)\ge\frac{1}{15}\left(1 - \frac{1}{e}\right) \left(f\left(B_{\numpartitions, \lastbucket}\right) - \frac{\tau}{2}\right),\]
where $B_{\numpartitions, \lastbucket}$ is any bucket in the last partition that is not saturated.
\end{lemma}
\begin{proof}
Let $\numpartitions=\ceil{\log K}$ and let $A_i$ be the bucket in partition $i$ with $c(A_i)<2^i$ that minimizes $c(A_i\cap E)$. 
We denote $E_i := A_i \cap E$. 
By setting $B = B_{\numpartitions, \lastbucket}$, $R=E_{\numpartitions}$ and $A = \bigcup_{i=0}^{\numpartitions-1}(A_i\setminus E_i)$ in \lemref{lemma:A-B-R}, it follows that $f\left(\bigcup_{i=0}^{\numpartitions}(A_i\setminus E_i)\right)\ge f\left(B_{\numpartitions, \lastbucket}\right) - 
f\left(E_{\numpartitions}\bigg|\bigcup_{i=0}^{\numpartitions-1}(A_i\setminus E_i)\right)$. 
Hence, applying \lemref{lem:recursion} with the observation that each $c(A_i)<2^i$, then
\begin{equation}
f\left(\bigcup_{i=0}^{\numpartitions}(A_i\setminus E_i)\right) \ge f\left(B_{\numpartitions, \lastbucket}\right) - \sum_{i = 1}^{\numpartitions} \frac{\tau}{2^{i - 1}}c(E_i)\eqnlab{eqn:B-union-apply-lemma-two}.
\end{equation}

Let $\tilde{E}_i$ denote the subset of $E$ that intersects buckets of partition $i$. 
For the sake of presentation, we denote $\beta:=8\numpartitions$. 
Then the total space in partition $i^*$ is at least $2^i\cdot w\ceil{K/2^i}+\beta\cdot c(\tilde{E}_i)$. 
Similarly, the number of buckets in partition $i^*$ is at least $w\ceil{K/2^i}+\frac{\beta\cdot c(\tilde{E}_i)}{2^i}$, of which at least half are not saturated. 
Since $A_i$ is defined to be the bucket of partition $i$ that minimizes $c(E_i)$ among all the buckets that are not saturated, then by an averaging argument
\[c(E_i)\le\frac{2c(\tilde{E}_i)}{w\ceil{K/2^i}+\frac{\beta\cdot c(\tilde{E}_i)}{2^i}}\le\frac{2c(\tilde{E}_i)}{w\ceil{K/2^i}+\frac{\beta\cdot c(\tilde{E}_i)}{2^i}}.\]
Therefore,
\begin{align*}
\sum_{i = 1}^{\numpartitions} \frac{\tau}{2^{i - 1}}c(E_i)&=\sum_{i = 1}^{\numpartitions} \frac{\tau}{2^{i - 1}}c(E_i)\\
&\le\left(\sum_{i = 1}^{\numpartitions}\frac{\tau}{2^{i - 1}}\frac{2^{i+1}c(\tilde{E}_i)}{wK+\beta\cdot c(\tilde{E}_i)}\right).
\end{align*}
Define $\alpha_i:=\frac{c(\tilde{E}_i)}{K}$ for $1\le i\le\numpartitions$. 
Then $\sum_{i = 1}^{\numpartitions} \frac{\tau}{2^{i - 1}}c(E_i)\le\left(\sum_{i = 1}^{\numpartitions}\frac{4\tau\alpha_i K}{wK+\beta\alpha_i K}\right)$. 
Hence, 
\begin{align}
\eqnlab{eqn:prejensen}
\sum_{i = 1}^{\numpartitions} \frac{\tau}{2^{i - 1}}c(E_i)\le\left(\sum_{i = 1}^{\numpartitions}\frac{4\tau\alpha_i}{w+\beta\alpha_i}\right).
\end{align}
Note that since $\sum_{i=1}^{\numpartitions}c(\tilde{E}_i)\le c(E)\le mK$, then $\sum_{i=1}^{\numpartitions}\alpha_i\le m$. 
Moreover, defining the function $f(x)$ by $f(x):=\frac{4x}{w+\beta x}$, where $w=\ceil{\frac{4\ceil{\log K}m}{K}}>0$, we see that $f(x)$ is concave. 
Thus by Jensen's inequality and setting $\alpha=\frac{m}{\numpartitions}$, it follows that $\left(\sum_{i = 1}^{\numpartitions}\frac{4\tau\alpha_i}{w+\beta\alpha_i}\right)\le\numpartitions\cdot\frac{4\tau\alpha}{w+\beta\alpha}$. 
Since $\beta=8\numpartitions$, then 
\begin{align}
\eqnlab{eqn:jensen}
\left(\sum_{i = 1}^{\numpartitions}\frac{4\tau\alpha_i}{w+\beta\alpha_i}\right)\le\numpartitions\frac{4\tau}{\beta}\le\frac{\tau}{2}.
\end{align}
Plugging \eqnref{eqn:prejensen} and \eqnref{eqn:jensen} into \eqnref{eqn:B-union-apply-lemma-two},
\[f\left(\bigcup_{i=0}^{\numpartitions}(A_i\setminus E_i)\right) \ge f\left(B_{\numpartitions, \lastbucket}\right) - \frac{\tau}{2}.\]
We can also bound the cost of the elements in $\bigcup_{i=0}^{\numpartitions}(A_i\setminus E_i)$:
$\bigcup_{i=0}^{\numpartitions} c(A_i \setminus E_i) \le \bigcup_{i=0}^{\numpartitions} c(A_i) \le 2K + 2K + \bigcup_{i=0}^{\lceil \log{K} \rceil-2}{2\cdot 2^i} \le 8K$. 
Hence, the optimal value of $f$ on $S_\tau\setminus E$ on a set of cost $8K$, which we denote $f\left(\OPT(8K,S_\tau\setminus E)\right)$, is at least $f\left(B_{\numpartitions, \lastbucket}\right) - \frac{\tau}{2}$.	
By \lemref{lem:numsets}, any set with size $8K$ whose items have size at most $K$ can be partitioned into $15$ sets, each with size at most $K$. 
Therefore by \thmref{thm:offline}, $f(Z_\tau) = f(\offline(K, S_\tau\setminus E)) \ge\left(1-\frac{1}{e}\right)f\left(\OPT(K,S\setminus E)\right)$. 
By submodularity, $f(Z_\tau)\ge \left(1 - \frac{1}{e}\right)\left( \frac{1}{11} f\left(\OPT(6K, S_\tau \setminus E)\right)\right)$ and thus, $f(Z_\tau)\ge\frac{1}{15}\left(1 - \frac{1}{e}\right) \left(f\left(B_{\numpartitions, \lastbucket}\right) - \frac{\tau}{2}\right)$.
\end{proof}

Finally, in the third case of our analysis, where more than half of the buckets in all partitions are not saturated and no bucket $B_{\numpartitions, \lastbucket}$ in the last partition produces a value $f(B_{\numpartitions, \lastbucket})$ that is a good estimate of $f(\OPT)$. 

\begin{lemma}
\lemlab{lem:items:bad}
Let $\numpartitions=\ceil{\log K}$ and $\tau>0$. 
If no partition in $S_\tau$ has at least half of its buckets saturated, then: $$f(Z_\tau)\ge\left(1-\frac{1}{e}\right)\left(f(\OPT)-f(B_{\numpartitions, \lastbucket})-\tau\right),$$ where $B_{\numpartitions, \lastbucket}$ is any bucket in the last partition that is not saturated.
\end{lemma}
\begin{proof}
Let $Y$ be the set that contains all elements from $\OPT(K,V\setminus E)$ that are buckets in $S_\tau$ with higher priority than $B_{\numpartitions, \lastbucket}$ and let $X := \OPT(K,V\setminus E)\setminus Y$. 
For each $e\in X$, 
\begin{equation}
\eqnlab{eqn:last-bucket}
\margain{e}{B_{\numpartitions, \lastbucket}} < \frac{\tau}{K}\cdot c(e) 
\end{equation}
due to the fact that $B_{\numpartitions, \lastbucket}$ is the bucket in the last partition and is not saturated.

Since $f(\OPT(K,V\setminus E))=f(X\cup Y)$, then by submodularity, $f(Y) \geq f(\OPT(K,V\setminus E)) - f(X)$. 
Then by monotonicity, $f(Y)\geq f(\OPT(K,V\setminus E)) - \margain{X}{B_{\numpartitions, \lastbucket}} - f\left(B_{\numpartitions, \lastbucket}\right)$. 
By submodularity, $f(Y)\geq f(\OPT(K,V\setminus E)) - f\left(B_{\numpartitions, \lastbucket}\right) - \sum_{e \in X} \margain{e}{B_{\numpartitions, \lastbucket}}$. 
Then by \eqnref{eqn:last-bucket}, $f(Y) \geq f(\OPT(K,V\setminus E)) - f\left(B_{\numpartitions, \lastbucket}\right) - \frac{\tau}{K}c(X)$. 
Since $c(X)\le K$, then 
\begin{align}
f(Y)\ge f(\OPT(K,V\setminus E)) - f\left(B_{\numpartitions, \lastbucket}\right) - \tau. \eqnlab{eqn:fy}
\end{align}
Therefore, by \thmref{thm:offline}, $f(Z_\tau) = f(\offline(K, S_\tau\setminus E)) \ge \left(1 - 1/e\right) f(K,S_\tau\setminus E)$. 
Since $Y\subseteq(S_\tau\setminus E)$, then $f(Z_\tau)\ge\left(1 - 1/e\right) f(\OPT(K,Y))$. 
The capacity of $Y$ is at most $K$, so $f(Z_\tau)\ge\left(1 - 1/e\right) f(Y)$. 
Hence by \eqnref{eqn:fy}, $f(Z_\tau)\ge\left(1 - 1/e\right)\left(f(\OPT(K,V\setminus E)) - f(B_{\numpartitions, \lastbucket}) - \tau\right)$, as desired.
\end{proof}

Since the above three lemmas hold for every $\tau > 0$  we can pick its value to give the desired approximation guarantee and space complexity bound. This gives \thmref{thm:items}, which corresponds to the first part of \thmref{thm:one-knapsack}.

\begin{theorem}
\thmlab{thm:alg:items}
Let $\numpartitions=\ceil{\log K}$ and $\zeta=1-\frac{1}{2\numpartitions}$. 
There exists an algorithm that outputs a set $\robust$ with $\tilde{O}(K^2+mK)$ elements such that, for any set $E$ of at most $m$ removed items, one can compute from $\robust$ a set $Z \subseteq V \setminus E$ with cost at most $K$ and 
$$f(Z)\ge\left(\frac{2 (1 - 1/e)\zeta}{32\zeta+3}-\eps\right) f(\OPT(V \setminus E)).$$ 
\end{theorem}
\begin{proof}
Fix any value of $\tau > 0$ and consider the bounds of \lemref{lem:items:good} and \lemref{lem:items:bad} as functions of $f(B_{\numpartitions, \lastbucket})$. 
Note that the first bound increases and the second bound decreases as a function of this parameter. 
Since we can always pick the better bound, the worst value for $f(B_{\numpartitions, \lastbucket})$ is when the two bounds are equal, i.e. $\frac{1}{15}\left(1-\frac{1}{e}\right)\left(f(B_{\numpartitions,\lastbucket})-\frac{\tau}{2}\right)=\left(1-\frac{1}{e}\right)\left(f(\OPT)-f(B_{\numpartitions, \lastbucket})-\tau\right)$. 
This occurs when $f(B_{\numpartitions, \lastbucket})=\frac{15f(\OPT)}{16}-\frac{29\tau}{32}$. 
Hence, the best of these two bounds is always at least $\frac{1}{16}\left(1 - \frac1e\right)(f(\OPT) - \frac{3 \tau}{2})$.

Note that this is a decreasing function of $\tau$ while the inequality of \lemref{lem:items:saturated} is an increasing function of $\tau$, so we pick $\tau$ to make sure that the minimum of these two bounds is large. 
The optimal value of $\tau$ is $\tau=\frac{2}{32\zeta+3}\cdot f(\OPT)$ in which case the two bounds are equal.  
Hence, $f(Z)\ge\frac{2\ (1 - 1/e)\zeta}{32\zeta+3}\cdot f(\OPT)$. 
By making guesses $\tau^*$ for $f(\OPT)$ by increasing powers of $(1+\eps)$, we can obtain a $(1+\eps)$ approximation of $\tau$, giving a $\left(\frac{2\ (1 - 1/e) \zeta}{32\zeta+3}-\eps\right)$ approximation for $f(\OPT)$. 

Finally, we give the bound on the number of elements returned. 
The number of buckets is dynamically updated until $\sum_{j=1}^{\numbuckets_i}|B_{i,j}|<10w\cdot 2^i$. 
Hence at most $80\cdot2^i\numpartitions w$ new buckets have been created for partition $i$. 
Then the total number of elements in each partition is at most $\sum_{i = 0}^\numpartitions\left(10\numpartitions w+8\numpartitions+80\cdot2^i\numpartitions w\right) 2^i = \O{K^2\numpartitions w} = \O{K^2\log K+mK\log^2 K}$ since $w = \ceil{\frac{4\numpartitions m}{K}}$. 
Since there are $\numpartitions=\O{\log K}$ partitions, the total number of elements is $\O{K^2\log^2 K+mK\log^3 K}$ for each guess of $\tau$. 
Assuming $\O{f(\OPT)}=\O{\log K}$, then the total number of guesses for $\tau$ is $\O{\frac{1}{\eps}\log K}$, so the total number of elements is $\O{\frac{1}{\eps}(K^2\log^3+mK\log^4 K)}$.
\end{proof}
We now show that \prune{} reduces the total number of elements output, while maintaining a constant factor approximation. 
\begin{lemma}
\lemlab{lem:prune}
Suppose \algnum{} outputs a set $\robust$ from which one can compute a set $Z\subseteq V\setminus E$ with cost at most $K$ and $f(Z)$ is an $r$-approximation to $f(\OPT)$. 
Then \prune{} outputs a set $T$ of size $\O{\frac{1}{\eps}(K\log^3+m\log^4 K)}$, from which one can compute a set $W\subseteq V\setminus E$ with cost at most $K$ and $f(W)$ is an $r^2$-approximation to $f(\OPT)$.
\end{lemma}
\begin{proof}
Since \prune{} runs an instance of \algnum{} on $S$, then \prune{} provides an $r$-approximation to $f(\OPT(K,S\setminus E)$, which is an $r$-approximation to $f(\OPT)$. 
Thus, $f(W)$ is an $r^2$-approximation to $f(\OPT)$.

It remains to bound the number of elements in $W$. 
Consider the state of \algnum{} on the set $S$ sorted by size. 
Note that no new buckets are created in partition $i$ when the total number of items in the bucket is $10w\cdot 2^i$. 
Let $u_i$ be the first time at which partition $i$ contains $10w\cdot2^i$ items, and let all the elements placed in partition $i$ before time $u_i$ be called ``old'' while all the elements that are placed in partition $i$ after time $u_i$ be called ``new''. 
Observe that each old element increments $\extraspace_i$ by a $8\numpartitions$ multiple of its cost.  
Since each new element costs at least as much as each old element, $80\numpartitions w\cdot 2^i$ new elements will cost at least $8\numpartitions$ times the cost of the old elements, which fills the additional space allocated by the old elements. 

Hence, the total number of elements in each partition is at most $\sum_{i = 0}^\numpartitions\O{\numpartitions w\cdot 2^i} = \O{K\numpartitions w} = \O{K\log K+m\log^2 K}$ since $w = \ceil{\frac{4\numpartitions m}{K}}$. 
Since there are $\numpartitions=\O{\log K}$ partitions, the total number of elements for each guess of $\tau$ is $\O{K\log^2 K+m\log^3 K}$. 
Assuming $\O{f(\OPT)}=\O{\log K}$, then the total number of guesses for $\tau$ is $\O{\frac{1}{\eps}\log K}$, so the total number of elements is $\O{\frac{1}{\eps}(K\log^3+m\log^4 K)}$.
\end{proof}
\noindent
Together, \thmref{thm:alg:items} and \lemref{lem:prune} give the proof of \thmref{thm:items}. 
A similar approach can be used to prove \thmref{thm:mult}.

\section{Missing Proofs from Section~\ref{sec:d-knapsacks}}
\label{app:missing:mult}
For a specific knapsack $a$, we call a bucket $B_{i,j}$ saturated with respect to knapsack $a$ if $c_a(B_{i,j})\ge\min(2^i,K)$. 
As before, we use $S_\tau := \{B_{i,j}\}_{i,j}$ to denote the data structure output by \algnum when run with parameter $\tau$  (i.e. for $\tau^* = \frac{\tau}{4})$ and $Z_\tau := \offline(\bigcup_{B_{i,j} \in S_\tau} B_{i,j}\setminus E)$, where $E$ is the set of elements that are removed at the end of the stream. 
Finally, recall that $\numpartitions=\ceil{\log K}$. 
\begin{lemma}
\lemlab{lem:mult:saturated}
Let $\tau>0$. 
For a fixed knapsack $a$, if there exists a partition in $S_\tau$ with at least half of its buckets saturated with respect to knapsack $a$, then 
\[f(Z_\tau)\ge\frac{1}{1+2d}\left(1-\frac{1}{e}\right)\left(1-\frac{1}{2\numpartitions}\right)\tau.\]
\end{lemma}
\begin{proof}
Let $a$ be a fixed knapsack and $i^*$ be a partition in $S_\tau$ with at least half of its buckets saturated with respect to knapsack $a$. 
Let $B_{i^*,j^*}$ be a saturated bucket that minimizes $c_a(B_{i^*,j^*}\cap E)$. 
Let $I$ be the cost of the items of $E$ with respect to knapsack $a$ that are in partition $i^*$,
\[I=c_a\left(\bigcup_{j=1}^{\numbuckets_{i^*}}(B_{i^*,j}\cap E)\right).\]
Then the total space in partition $i$ is at least $2^{i^*+1}\cdot w\ceil{K/2^{i^*}}+(8\numpartitions)I$, so the total number of buckets is at least $\frac{2^{i^*+1}\cdot w\ceil{K/2^{i^*}}+(8\numpartitions)I}{2^{i^*}}$. 
Since at least half of its buckets are saturated with respect to knapsack $a$, the total number of saturated buckets is at least $\frac{\capacity_i\cdot w\ceil{K/2^{i^*}}+(4\numpartitions)I}{2^{i^*}}$

By an averaging argument, the cost of the elements in $B_{i^*,j^*}$ with respect to knapsack $a$ that are removed by $E$ is at most
\[c_a(B_{i^*,j^*}\cap E)\le\frac{2^{i^*}\cdot I}{2^{i^*}\cdot w\ceil{K/2^{i^*}}+(4\numpartitions)I}.\]
Thus, $c_a(B_{i^*,j^*} \setminus E)$ is at least $c_a(B_{i^*,j^*})-\frac{2^{i^*}\cdot I}{2^{i^*}\cdot w\ceil{K/2^{i^*}}+(4\numpartitions)I}$.

Note that if $B_{i^*,j^*}$ is saturated with respect to knapsack $a$, then the marginal density of each element exceeds a threshold of $\frac{\tau}{(1+2d)\cdot c_a(B_{i^*,j^*})}$ so that $f(B_{i^*,j^*} \setminus E)\ge\left(c_a(B_{i^*,j^*})-\frac{2^{i^*}\cdot I}{2^{i^*}\cdot w\ceil{K/2^{i^*}}+(4\numpartitions)I}\right)\cdot\frac{\tau}{(1+2d)\cdot c_a(B_{i^*,j^*})}$. 
Since $\frac{2^{i^*}\cdot I}{\capacity_i\cdot w\ceil{K/2^{i^*}}+(4\numpartitions)I}<\frac{2^{i^*}}{4\numpartitions}$ and $c_a(B_{i^*,j^*})\ge 2^{i^*-1}\ge 1$ for a saturated bucket $B_{i^*,j^*}$, then it follows that $f(B_{i^*,j^*} \setminus E)\ge\left(1-\frac{1}{2\numpartitions}\right)\frac{\tau}{1+2d}$. 
Hence by \thmref{thm:offline}, running \offline{} on $B_{i^*,j^*} \setminus E$ produces a $\frac{1}{1+2d}\left(1-\frac{1}{e}\right)\left(1-\frac{1}{2\numpartitions}\right)$ approximation.
\end{proof}
The following lemma corresponds to \lemref{lem:Ei}, using the threshold of \algmult.
\begin{lemma}
\lemlab{lem:mult:Ei}
Let $E_i:=A_i\cap E$ denote the elements that are removed from a bucket $A_i$ in partition $i>0$. 
Given a bucket $A_{i-1}$ from partition $i-1$ that is not saturated and any knapsack $a$, then the loss in bucket $A_i$ induced by the removals is at most
\[\margain{E_i}{A_{i-1}}<\frac{\tau}{2^{i-1}(1+2d)}c_a(E_i).\]
\end{lemma}
The following lemma corresponds to \lemref{lem:recursion}, using \lemref{lem:mult:Ei} and the threshold of \algmult.
\begin{lemma}
\lemlab{lem:mult:recursion}
Suppose that there exists some bucket in every partition that is not saturated with respect to some particular knapsack $a$.  
For every partition $i$, let $A_i$ denote a bucket with $c_a(A_i)<\min\{2^i,K\}$ and let $E_i:=A_i \cap E$ denote the elements that are removed from $A_i$. 
The loss in the bucket $B_{\numpartitions, \lastbucket}$ induced by the removals, given the remaining elements in the previous buckets, is at most $\margain{E_{\numpartitions}}{\bigcup_{j = 0}^{\numpartitions - 1}\left(A_j \setminus E_j \right)} \le \sum_{j = 1}^{\numpartitions} \frac{\tau}{2^{j - 1}(1+2d)} c_a(E_j)$.
\end{lemma}
The following lemma corresponds to \lemref{lem:items:good}, using \lemref{lem:mult:recursion}, the threshold of \algmult, and the observation that an optimal solution considering only a particular knapsack constraint is at least as good as an optimal solution considering additional other knapsack constraints. 
However, the $16d$ factor in the denominator results from a bucket $c(A_i)\le d\cdot 2^i$, due to the definition of $c(e)=\max_{1\le a\le d} c_a(e)$.
\begin{lemma}
\lemlab{lem:mult:good}
Let $\tau>0$. 
If no partition in $S_\tau$ has at least half of its buckets saturated with respect to any knapsack $a$ with $1\le a\le d$, then  $f(Z_\tau)\ge\frac{1}{16d}\left(1 - \frac{1}{e}\right)\cdot\left(f\left(B_{\numpartitions, \lastbucket}\right) - \frac{\tau}{2(1+2d)}\right)$, where $B_{\numpartitions, \lastbucket}$ is any bucket in the last partition that is not saturated.
\end{lemma}
The following lemma is similar to \lemref{lem:items:bad} and follows along the same proof, with the observation that $c(X)\le dK$. 
\begin{lemma}
\lemlab{lem:mult:bad}
Let $\tau>0$. 
If no partition in $S_\tau$ has at least half of its buckets saturated with respect to any knapsack $a$ with $1\le a\le d$, then $f(Z_\tau)\ge\left(1-\frac{1}{e}\right)\cdot\left(f(\OPT(K,V\setminus E))-f(B_{\numpartitions, \lastbucket})-\frac{d\tau}{1+2d}\right)$, where $B_{\numpartitions, \lastbucket}$ is any bucket in the last partition that is not saturated.
\end{lemma}
We now prove \thmref{thm:mult}. 
\begin{proofof}{\thmref{thm:mult}}
The $r$-approximation guarantee follows from \lemref{lem:mult:saturated}, \lemref{lem:mult:good} and \lemref{lem:mult:bad}, when $f(B)=\frac{f(\OPT)}{2}$, and $\tau=\frac{f(\OPT)}{4}$. 
The $r^2$-approximation guarantee and space bounds follow from \lemref{lem:prune} with \prune{} using \algmult{} instead of \algnum.
\end{proofof}

\section{Missing Proofs from Section~\ref{sec:distributed}}
\label{app:missing:distributed}
We first prove \lemref{lemma:distributed-approximation}.
\begin{proofof}{\lemref{lemma:distributed-approximation}}
We will show that $S$ returned by line~\ref{line:distributed-return-S} of \algref{alg:distributed} equals to the output of \algmult run on a stream of $V$ such that:
\begin{itemize}
\item $F$ is a prefix of this stream.
\item Elements of $F \cup R$ appear in the same order in the stream as they appear in \algref{alg:distributed}.
\item The order of the remaining elements is arbitrary.
\end{itemize}
Let $\cB_{\algmult}$ be the structure of sets $B_{i, j}$ and their content after $\algmult$ is executed on this stream.
	
Next, recall that $\algmult$ never removes any element from any $B_{i, j}$. Also, recall that $\cB_0$ is obtained by executing $\algmult$ on $F$. 
Hence, since $F$ is a prefix of the stream, for each $(i, j)$ the content of $B_{i, j}$ in $\cB_0$ is a subset of the content of $B_{i, j}$ in $\cB_{\algmult}$. 
Furthermore, $f$ is a submodular function, so if an element $e$ is not added to $B_{i, j}$ due to its small marginal gain, $e$ will not be added to a superset of $B_{i, j}$ neither. This implies that after $F$ is processed, no element $e$ that is not added to $R_i$ on line~\ref{alg:distributed-element-e} of \algref{alg:distributed} can ever be added to any $B_{i, j}$ (regardless of ordering of the elements $V \setminus F$). Therefore, the only relevant elements in the rest of the stream are those in $R$.
	
The proof now follows from the fact that \thmref{thm:mult} holds regardless of ordering of the stream.
\end{proofof}
To prove \lemref{lemma:distributed-memory-bound}, we need the following result for submartingales.
\begin{theorem}[Azuma's Inequality]
\thmlab{thm:azuma}
Suppose $X_0, X_1,\ldots,X_n$ is a submartingale and $|X_i-X_{i+1}|\le c_i$. 
Then
\[\PPr{X_n-X_0\le -t}\le\exp\left(\frac{-t^2}{2\sum_i c_i^2}\right).\]
\end{theorem}
\begin{proofof}{\lemref{lemma:distributed-memory-bound}}
Observe that the expected number of elements in $S$ is $4 \sqrt{n \memory}$ so that $|S| < 3 \sqrt{n \memory}$ occurs only with probability at most $e^{-\Omega(\sqrt{n \memory})} \leq e^{-\Omega(\memory)}$ by standard Chernoff bounds. 
Thus, $|S| \geq 3 \sqrt{n \memory}$ with high probability. 
Let $N_S$ denote the total number of elements $e$ that are added to $\cB_0$ by $\algmult_{\cB_0, \{e\}}(d, m, K, \tau)$, so that exactly $N_S+|S|$ elements are sent to $C$ in round two. 

Suppose we split the sample set $S$ into $3 \memory$ pieces of size $\sqrt{n/\memory}$ and process each piece sequentially. 
Suppose further that before some piece, there are at least $\sqrt{n \memory}$ remaining elements that would be added to $\cB_0$ as described above. 
Then an additional element is added to $\cB_0$ with probability at least $1 - \left(1 - \sqrt{\frac{\memory}{n}}\right)^{\sqrt{\frac{n}{\memory}}} > 1/2 $, conditioned on \emph{any} previous actions of the algorithm, since each piece can be sampled independently and thus we can use a martingale argument to bound the number of elements selected in $S$. 

Let $X_i$ be the indicator random variable for the event that at least one element is selected from the $i$-th piece so that $\mathbb{E}[X_i \mid X_1,\ldots,X_{i-1}] \geq 1/2$. 
Let $Y_i = \sum_{j=1}^{i} (X_i - 1/2)$ so that the sequence $Y_1,Y_2,\ldots$ is a submartingale and hence, $\mathbb{E}[Y_i \mid Y_1,\ldots,Y_{i-1}] \geq Y_{i-1}$ and $|Y_{i} - Y_{i-1}| \leq 1$. 
Therefore, $\Pr[Y_{3\memory} < -\frac12 \memory] < e^{-\Omega(\memory)}$ by Azuma's inequality (\thmref{thm:azuma}). 
Hence with probability $1-e^{-\Omega(\memory)}$, $\sum_{j=1}^{3\memory} X_j = Y_{\memory} + \frac{3}{2} \memory \geq \memory$ and $\cB_0$ includes at least $\memory$ elements overall, in which case nothing is sent to the central machine. 
Otherwise, the number of remaining elements added to $\cB_0$ is less than $\sqrt{n \memory}$. 
\end{proofof}

\section{Robust to removal of size $M$}
\applab{app:size}
In this section, we consider the ARMSM$(m,K)$ problem under a single knapsack constraint, when the $m$ items have cost at most $M$. 
In contrast to \algnum, we no longer need a dynamic allocation of new buckets, so having a fixed number of buckets for each partition suffices. 
We give our algorithm in full in \algsize.
\begin{algorithm}
\caption{\algsize: Picking elements with large marginal gain to cost ratio.}
\begin{algorithmic}[1]
\Require{Parameters $M$, $K$, estimate $\tau$ of $f(\OPT)$.}
\State{$\numpartitions \leftarrow \ceil{\log K}, w \leftarrow \ceil{\frac{4\numpartitions M}{K}}$, $\tau\leftarrow\frac{f(\OPT)}{13-11\left(\frac{4M}{wK}\right)}$}
\For{$i \leftarrow 0$ to $\numpartitions$}
\Comment{Initialize parameters}
\State{$\numbuckets_i\leftarrow w\ceil{K/2^i}$}
\State{$\capacity_i = \min\{2^i, K\}$}
\For{$j \leftarrow 1$ to $\numbuckets_i$}
\State{$B_{i,j}\leftarrow\emptyset$} 
\EndFor
\EndFor
\For{each element $e$ in the stream}
\For{$i\leftarrow 0$ to $\numpartitions$}
\If {$c(e) > \min\{2^{i-1}, K\}$} \textbf{continue}
\EndIf
\For{$j\leftarrow 1$ to $\numbuckets_i$}
\If{$\rho(e | B_{i,j}) < \frac{\tau}{\capacity_i}$} \textbf{continue} 
\EndIf
\If{$c(B_{i,j} \cup e)< 2 \capacity_i$ }
\State{$B_{i,j}\leftarrow B_{i,j}\cup\{e\}$}
\State{\textbf{break:} process next element $e$}
\EndIf
\EndFor
\EndFor
\EndFor\\
%\Return $S_\tau = \{B_{i,j}\}_{i,j}, Z_\tau = \offline(\bigcup_{i,j} B_{i,j}\setminus E)$
\Return $Z_\tau = \offline(\bigcup_{i,j} B_{i,j}\setminus E)$
\end{algorithmic}
\end{algorithm}
To show that the optimal solution of $Z$ output by \algsize{} is a good approximation to the optimal solution of the entire stream, we call a bucket $B_{i,j}$ saturated if $c(B_{i,j})\ge\min\{2^i,K\}$ and break the analysis into the following three cases:
\begin{enumerate}
\item
At least half of the buckets in some partition are saturated (\lemref{lem:size:saturated})
\item
More than half of the buckets in all partitions are not saturated, but there exists some bucket in the last partition that is a good estimate of $f(\OPT)$ (\lemref{lem:size:good})
\item
More than half of the buckets in all partitions are not saturated and no bucket in the last partition is a good estimate of $f(\OPT)$ (\lemref{lem:size:bad})
\end{enumerate}
In the first case, if most of the buckets in some partition are saturated, we argue through an averaging argument that some saturated bucket $B_{i,j}$ in this partition cannot have too much size intersection with the elements $E$ that are removed at the end of the stream, giving a lower bound on $c(B_{i,j} \setminus E)$.
Since elements can only be added to this bucket if the ratio of their marginal gain to their size exceeds a certain threshold, then we conclude that $f(B_{i,j} \setminus E)$ is at least the product of $c(B_{i,j} \setminus E)$ this threshold, which gives a good approximation to $f(\OPT)$. 

In the second case, if there exists some bucket $B_{\numpartitions,\lastbucket}$ in the last partition that is a good estimate of $f(\OPT)$, we first use a technical lemma to show that the optimal solution on $Z$ is at least $f(B_{\numpartitions,\lastbucket})$ minus the value of the elements across all the buckets that were deleted by $E$. 
To bound the value of these elements, we argue that if most of the buckets in all partitions are not saturated, then no element in a bucket $B_{i,j}$ that is deleted by $E$ can value that is too high, because otherwise it would have been added to a bucket in some previous partition less than $i$. 
Hence, we derive an upper bound on the value of the elements across all the buckets that were deleted by $E$, and this suffices to show that the optimal solution of $Z$ is close to $f(\OPT)$, since $f(B_{\numpartitions,\lastbucket})$ is a good approximation to $f(\OPT)$. 

In the third case, if all buckets in the last partition give poor estimates of $f(\OPT)$, then for each of these buckets, the total size of the elements in the bucket cannot be large. 
As a result, most elements of $\OPT$ must either be in a previous partition or have poor marginal gain. 
If most elements of $\OPT$ are in a previous partition, then the union of the items in the previous partitions are contained in $Z$ and thus the optimal solution of $Z$ is a good approximation to $f(\OPT)$. 
If most elements of $\OPT$ have poor marginal gain, then there must be some item of $\OPT$ with substantial value. 
On the other hand, since each partition contain many buckets that are not saturated, then this substantial item must have been captured by some bucket in a previous partition and so again, the optimal solution of $Z$ is a good approximation to $f(\OPT)$. 
Intuitively, if at least half of the buckets in some partition are saturated, then some saturated bucket $B_{i^*,j}$ in this partition cannot be affected too much by the removal of elements at the end of the stream. 
Hence, this bucket $B_{i^*,j}$ gives a good approximation to $\tau$, which in turn serves as a good approximation to $f(\OPT)$.  
\begin{lemma}
\lemlab{lem:size:saturated}
Let $\tau>0$.  
If there exists a partition in $S_\tau$ such that at least half of its buckets are saturated, then 
\[f(Z_\tau)\ge\left(1-\frac{1}{e}\right)\left(1-\frac{4M}{wK}\right)\tau.\]
\end{lemma}
\begin{proof}
Let $i$ be a partition such that half of its buckets are saturated. 
Let $B_{i,j}$ be a saturated bucket that minimizes $c(B_{i,j}\cap E)$. 
Since every partition contains $w\ceil{K/2^i}$ buckets, the number of saturated buckets in partition $i$ is at least $wK/2^{i+1}$. 
By a simple averaging argument and the observation that $c(E)\le M$,
\[c(B_{i,j}\cap E)\le\frac{2^{i+1}M}{wK}.\]
Thus,
\[c(B_{i,j} \setminus E)\ge c(B_{i,j})-\frac{2^{i+1}M}{wK}.\]
Note that if $B_{i,j}$ is saturated, then the marginal gain to weight ratio of each element exceeds a threshold of $\frac{\tau}{c(B_{i,j})}$ so that
\begin{align*}
f(B_{i,j} \setminus E)&\ge\left(c(B_{i,j})-\frac{2^{i+1}M}{wK}\right)\frac{\tau}{c(B_{i,j})}\\
&\ge\tau\left(1-\frac{4M}{wK}\right),
\end{align*}
where the last step follows from the observation that $c(B_{i,j})\ge 2^{i-1}$ for a saturated bucket $B_{i,j}$. 
Hence, running \offline{} on $B_{i,j} \setminus E$ produces a $\left(1-\frac{1}{e}\right)\left(1-\frac{4M}{wK}\right)\tau$ approximation by \thmref{thm:offline}.
\end{proof}

The second case of our analysis occurs when more than half of the buckets in all partitions are not saturated, but there exists some bucket in the last partition that is a good estimate of $f(\OPT)$. 
We now show that \algsize{} yields a good approximation in this case. 
\begin{lemma}
\lemlab{lem:size:good}
Let $\numpartitions=\ceil{\log K}$ and $\tau>0$. 
If no partition in $S_\tau$ has at least half of its buckets saturated, then 
\[f(Z_\tau)\ge\frac{1}{11}\left(1-\frac{1}{e}\right)\left(f(B_{\numpartitions, \lastbucket})-\frac{4M}{wK}\tau\right),\]
where $B_{\numpartitions, \lastbucket}$ is any bucket in the last partition that is not saturated.
\end{lemma}
\begin{proof}
Let $\capacity_i=\min\{2^i,K\}$. 
Let $B_i$ denote the bucket in partition $i$ with $c(B_i)<\capacity_i$ for which $c(E_i)$ is minimized, where $E_i := B_i \cap E$. 
By setting $B = B_{\numpartitions, \lastbucket}$, $R=E_{\numpartitions}$ and $A = \bigcup_{i=0}^{\numpartitions-1}(B_i\setminus E_i)$ in \lemref{lemma:A-B-R}, we have that $f\left(\bigcup_{i=0}^{\numpartitions}(B_i\setminus E_i)\right)\ge f\left(B_{\numpartitions, \lastbucket}\right) - 
f\left(E_{\numpartitions}\bigg|\bigcup_{i=0}^{\numpartitions-1}(B_i\setminus E_i)\right)$. 
Since each $c(B_i)<\capacity_i$, then by \lemref{lem:recursion}, $f\left(\bigcup_{i=0}^{\numpartitions}(B_i\setminus E_i)\right)$ is at least
\begin{align}
f\left(B_{\numpartitions, \lastbucket}\right) - \sum_{i = 1}^{\numpartitions} \frac{\tau}{2^{i - 1}}c(E_i)\eqnlab{eqn:B-union-apply-lemma}.
\end{align}
Let $\tilde{E}_i$ denote the subset of $E$ that intersects buckets of partition $i$. 
Since each $B_i$ is defined to be the bucket of partition $i$ that minimizes $c(E_i)$ among all the buckets that are not saturated and each partition $i$ contains $w\ceil{K/2^i}$ buckets, of which more than half are not saturated, then $c(E_i)$ is at most
\[c(E_i)\le\frac{c(\tilde{E}_i)}{\frac{w\ceil{K/2^i}}{2}} \le \frac{2^{i + 1} c(\tilde{E}_i)}{wK}.\]
Hence,
\begin{align*}
\sum_{i = 1}^{\numpartitions} \frac{\tau}{2^{i - 1}}c(E_i) &\le \sum_{i = 1}^{\numpartitions} \frac{\tau}{2^{i - 1}} \frac{2^{i + 1} c(\tilde{E}_i)}{wK}\\
&= \frac{4}{wK} \tau \sum_{i = 1}^{\numpartitions} c(\tilde{E}_i)\\
&\le \frac{4c(E)}{wK}\tau\le\frac{4M}{wK}\tau.
\end{align*}
Plugging this inequality into \eqnref{eqn:B-union-apply-lemma},
\[f\left(\bigcup_{i=0}^{\numpartitions}(B_i\setminus E_i)\right) \ge f\left(B_{\numpartitions, \lastbucket}\right) - \frac{4M}{wK} \tau.\]
The cost of the elements in $\bigcup_{i=0}^{\numpartitions}(B_i\setminus E_i)$ is at most $\bigcup_{i=0}^{\numpartitions} c(B_i) \le 2K + 2K+\bigcup_{i=0}^{\lceil \log{2K} \rceil-2}{2\cdot 2^i} \le 6K$. 
Hence, the optimal value of $f$ on $S_\tau\setminus E$ on a set of cost $6K$, which we denote $f\left(\OPT(6K,S_\tau\setminus E)\right)$, is at least $f\left(B_{\numpartitions, \lastbucket}\right) - \frac{4M}{wK} \tau$.	
Therefore,
\begin{align*}
f(Z_\tau) &= f(\offline(K, S_\tau\setminus E))\\
&\ge\left(1-\frac{1}{e}\right)f\left(\OPT(K,S_\tau\setminus E)\right)\\
&\ge \left(1 - \frac{1}{e}\right)\left( \frac{1}{11} f\left(\OPT(6K, S_\tau\setminus E)\right)\right)\\
&\ge \frac{1}{11}\left(1 - \frac{1}{e}\right) \left(f\left(B_{\numpartitions, \lastbucket}\right) - \frac{4M}{wK}\tau\right),
\end{align*}
where the first inequality holds by \thmref{thm:offline} and the second inequality holds by submodularity and the observation that any set with size $6K$ whose items have size at most $K$ can be partitioned into $11$ sets, each with size at most $K$ (i.e., \lemref{lem:numsets}). 
\end{proof}

\noindent
The following lemma is the same as \lemref{lem:items:bad}, with the identical proof.
\begin{lemma}
\lemlab{lem:size:bad}
Let $\numpartitions=\ceil{\log K}$ and $\tau>0$. 
If no partition in $S_\tau$ has at least half of its buckets saturated, then 
\[f(Z)\ge\left(1-\frac{1}{e}\right)\left(f(\OPT(K,V\setminus E))-f(B_{\numpartitions, \lastbucket})-\tau\right),\]
where $B_{\numpartitions, \lastbucket}$ is any bucket in the last partition that is not saturated.
\end{lemma}
\thmref{thm:size} then follows from optimizing parameters to give an approximation guarantee for an algorithm using $\algsize$, when the $m$ items have cost at most $M$. 
\begin{theorem}
\thmlab{thm:size}
For the ARMSM$(m,K)$ problem subject to a knapsack constraint, there exists an algorithm that outputs a set $Z$ so that $f(Z)$ is a constant factor approximation to $f(\OPT)$ and stores $\tilde{O}(K+M)$ elements, if the removed items have cost at most $M$.
%$\left(\frac{1-1/\log K}{13}-\eps\right)$-approximation to $f(\OPT)$ and stores $\O{\frac{1}{\eps}K\log^2 K + M\log^3 K)}$ elements. 
%As $K\rightarrow\infty$, the approximation ratio guarantee tends to approximately $0.0769-\eps$. 
%If oracle access to $f$ is given, this algorithm can be implemented in a streaming setting using $\tilde O(K + M)$ words of space.
\end{theorem}
\begin{proof}
Let $\eta=\frac{4M}{wK}$ and $f(\OPT):=f(\OPT(K,V\setminus E)$. 
Let $\capacity_i=\min\{2^i,K\}$. 
From \lemref{lem:size:good} and \lemref{lem:size:bad}, it follows that the worst case bound for $f(B_{\numpartitions, \lastbucket})$ occurs when $\frac{1}{11}\left(1-\frac{1}{e}\right)\left(f(B_{\numpartitions, \lastbucket})-\eta\tau\right)=\left(1-\frac{1}{e}\right)\left(f(\OPT)-f(B_{\numpartitions, \lastbucket})-\tau\right)$. 
Some straightforward computation shows this occurs when $f(B_{\numpartitions, \lastbucket})=\frac{11f(\OPT)}{12}-\frac{11\tau}{12}+\frac{\eta}{12}$. 
It follows from \lemref{lem:size:saturated} that the optimal value of $\tau$ occurs at $\tau=\frac{f(\OPT)}{13-11\eta}$, which gives $f(Z)\ge\left(1-\frac{1}{e}\right)\frac{1-\eta}{13-11\eta}f(\OPT)$. 
For $w=\ceil{\frac{4\numpartitions M}{K}}$, we have $\eta\le\frac{1}{\numpartitions}$, so $f(Z)\ge\frac{1-1/\log K}{13}f(\OPT)$. 
By making guesses for $f(\OPT)$ by increasing powers of $(1+\eps)$, we can obtain a $(1+\eps)$ approximation of $\tau$, giving a $\frac{1-1/\log K}{13}-\eps$ approximation for $f(\OPT)$. 

Allowing each element in $S$ to be stored by \algsize{} using one word of space, then the total space that \algsize{} uses is at most
\begin{align*}
|S|&\le\sum_{i=0}^{\numpartitions}w\ceil{K/2^i}\cdot2\capacity_i\\
&\le\sum_{i=0}^{\numpartitions}w\ceil{K/2^i}2^{i+1}\\
&\le (w\cdot 4K)(\log 2K+2).
\end{align*}
Since $w=\ceil{\frac{4\numpartitions M}{K}}$, then $|S|=\O{K\log K+M\log^2 K}$. 
Assuming $\O{f(\OPT)}=\O{\log K}$, then the total number of guesses for $\tau$ is $\O{\frac{1}{\eps}\log K}$. 
Therefore, the total number of stored elements is $\O{\frac{1}{\eps}(K\log^2 K+M\log^3 K)}$.
\end{proof}
\end{document}